\documentclass{IEEEtran}
\usepackage[dvipdfmx]{graphicx}
\usepackage{amsmath, amssymb, bm}
\usepackage{subfigure}
\usepackage[outdir=fig/]{epstopdf}
\usepackage{amsthm}
\usepackage{cases}
\usepackage{cite}
\usepackage{enumitem}
\usepackage{mathtools}
\usepackage{multirow}
\usepackage{bbm}
\usepackage{color}
\usepackage{ascmac}
\usepackage{arydshln}
\usepackage[table]{xcolor}

\newtheorem{theorem}{Theorem}

\newtheorem{corollary}{Corollary}
\newtheorem{definition}{Definition}

\newcommand{\la}{\lambda}
\newcommand{\li}{\lambda_i}

\newcommand{\boldx}{\bm{x}}
\newcommand{\boldxhat}{\hat{\bm{x}}}

\newcommand{\bmc}{\bm{c}}
\newcommand{\bmd}{\bm{d}}
\newcommand{\bS}{\mathbf{S}}
\newcommand{\bH}{\mathbf{H}}
\newcommand{\bW}{\mathbf{W}}
\newcommand{\bD}{\mathbf{D}}
\newcommand{\bU}{\mathbf{U}}
\newcommand{\bL}{\mathbf{L}}
\newcommand{\bA}{\mathbf{A}}
\newcommand{\bG}{\mathbf{G}}
\newcommand{\bI}{\mathbf{I}}
\newcommand{\bV}{\mathbf{V}}
\newcommand{\bLam}{\bm{\Lambda}}
\newcommand{\mA}{\mathcal{A}}

\newcommand{\mH}{\mathcal{H}}
\newcommand{\mS}{\mathcal{S}}
\newcommand{\mT}{\mathcal{T}}
\newcommand{\mV}{\mathcal{V}}
\newcommand{\mW}{\mathcal{W}}

\newcommand{\mB}{\mathcal{B}}

\DeclareMathOperator*{\argmin}{arg\,min}

\begin{document}
\setlist[description]{font=\normalfont}

\title{Generalized Sampling on Graphs With\\Subspace and Smoothness Priors}
\author{Yuichi Tanaka and Yonina C. Eldar
\thanks{Y. Tanaka is with the Graduate School of BASE, Tokyo University of Agriculture and Technology, Koganei, Tokyo 184--8588, Japan. Y. Tanaka is also with PRESTO, Japan Science and Technology Agency, Kawaguchi, Saitama 332--0012, Japan (email: \mbox{ytnk@cc.tuat.ac.jp}).}
\thanks{Y. C. Eldar is with Faculty of Mathematics and Computer Science, The Weizmann Institute of Science, Rehovot 7610001, Israel (email: \mbox{yonina.eldar@weizmann.ac.il}).}
\thanks{Yuichi Tanaka was partially funded by JST PRESTO under grants JPMJPR1656 and JPMJPR1935, and JSPS KAKENHI under Grant 19K22864.}
}

\markboth{}{}
\maketitle

\begin{abstract}
We propose a framework for generalized sampling of graph signals that parallels sampling in shift-invariant (SI) subspaces. This framework allows for arbitrary input signals, which are not constrained to be bandlimited. Furthermore, the sampling and reconstruction filters may be different. We present design methods of the correction filter that compensate for these differences and lead to closed form expressions in the graph frequency domain. In this study, we consider two priors on graph signals: The first is a subspace prior, where the signal is assumed to lie in a periodic graph spectrum (PGS) subspace. The PGS subspace is proposed as a counterpart of the SI subspace used in standard sampling theory. The second is a smoothness prior that imposes a smoothness requirement on the graph signal. We suggest the use of recovery techniques for when the recovery filter can be optimized and under a setting in which a predefined filter must be used. Sampling is performed in the graph frequency domain, which is a counterpart of ``sampling by modulation'' used in SI subspaces. We compare our approach with existing sampling techniques on graph signal processing. The effectiveness of the proposed generalized sampling approach is validated numerically through several experiments.
\end{abstract}

\section{Introduction}
Sampling theory for graph signals has been recently studied with the goal of building parallels of sampling results in standard signal processing \cite{Anis2016, Chen2015, Marque2016, Chepur2018, Puy2018, Tsitsv2016, Valses2018, Pesens2008, Pesens2010, Sakiya2019a}. 
Since the pioneering Shannon--Nyquist sampling theorem \cite{Shanno1949, Jerri1977}, sampling theories that encompass more general signal spaces beyond that of bandlimited signals in shift-invariant (SI) spaces have been widely studied with many promising applications \cite{Eldar2009, Eldar2003, Eldar2006, Unser2000, Unser1994, Eldar2004, Eldar2015}. More relaxed priors have also been considered such as smoothness priors. These theories allow for sampling and recovery of signals in arbitrary subspaces using almost arbitrary sampling and recovery kernels. These results are particularly useful in the SI setting in which sampling and recovery reduce to simple filtering operations.

Graph signal processing (GSP) \cite{Shuman2013, Ortega2018} is a relatively new field of signal processing that studies discrete signals defined on a graph. Recent work on GSP ranges from theory to practical applications including wavelet/filter bank design \cite{Hammon2011, Narang2013, Tanaka2014a, Sakiya2019}, learning graphs from observed data \cite{Dong2016, Kalofo2016, Egilme2017, Yamada2019}, restoration of graph signals \cite{Onuki2016, Ono2015}, image/point cloud processing \cite{Cheung2018}, and deep learning on graphs \cite{Bronst2017}.

One of the topics of interest in GSP is graph sampling theory \cite{Anis2016, Chen2015, Marque2016, Chepur2018, Puy2018, Tsitsv2016, Valses2018, Pesens2008, Pesens2010, Sakiya2019a}, which is aimed at recovering a graph signal from its sampled version. Most studies on sampling of graph signals have considered recovery of discrete graph signals from their sampled version \cite{Anis2016, Chen2015, Marque2016, Chepur2018, Puy2018, Tsitsv2016, Valses2018, Sakiya2019a}. Current approaches generally rely on vertex sampling. The graph can be highly irregular, namely, the number of edges connected to a vertex may vary significantly. Hence, the ``best" sampling set depends on the graph and assumed signal model; sampling set selection with different models of signals or features has been studied extensively in sensor networks and machine learning \cite{Cressi1993, Shewry1987, Krause2008, Sharma2015}. Graph sampling theory typically assumes that the signal is smooth on the graph, in which the smoothness is frequently measured based on the number of nonzero coefficients in the graph Fourier spectrum \cite{Pesens2008, Anis2016}, allowing for perfect recovery, which will differ depending on the sampling set.

Vertex domain sampling parallels sampling of the (discrete) time domain signals in standard signal processing. However, the sampled spectrum of a graph signal does not preserve its original shape. Whereas time domain sampling has a corresponding frequency (i.e., DFT) domain representation that preserves the shape of the spectrum (up to possible aliasing) \cite{Vaidya1993, Eldar2015, Vetter2014}, vertex domain sampling does not have such a simple relationship. Instead, we propose to build an analog of standard sampling in the \textit{graph frequency domain}. Graph frequency sampling has been proposed in \cite{Tanaka2018}. Vertex and graph frequency domain sampling approaches coincide under certain conditions (see Section \ref{subsec:bpt}).

Here, we expand on \cite{Tanaka2018} by building a \textit{generalized graph sampling} framework that allows for (perfect) recovery of graph signals beyond bandlimited signals, and parallels SI sampling for time domain signals. In SI sampling, the input subspace has a particular SI structure. Sampling is modeled by uniformly sampling the output of the signal convolved with an arbitrary sampling filter. Under a mild condition on the sampling filter, recovery is obtained using a correction filter having an explicit closed-form frequency response. Herein, we demonstrate how one can extend these ideas to graphs by defining an appropriate input space of graph signals and sampling in the graph frequency domain \cite{Tanaka2018}. In addition, our generalized sampling framework enables recovery of non-bandlimited graph signals from \textit{vertex domain sampling} for bipartite graphs by applying a relationship between vertex and graph frequency domain sampling. In contrast to graph filter bank approaches \cite{Narang2012, Narang2013, Sakiya2014a, Sakiya2016a, Tanaka2014a, Tanaka2019a, Trembl2016}, our proposed generalized framework only requires one branch (i.e., channel) to recover the full-band graph signals.

Our framework relies on graph sampling performed in the graph frequency domain \cite{Tanaka2018} as a counterpart to ``sampling by modulation'' under the SI setting \cite{Eldar2015, Vetter2014}. This sampling method maintains the shape of the graph spectrum. Unlike in SI sampling, in which sampling in the time domain coincides with that in the frequency domain, under the graph setting, vertex domain sampling and frequency domain sampling are generally different. Sampling by modulation enables a generalized graph sampling framework that is analogous to SI sampling, exhibiting a symmetric structure in which the sampling and reconstruction steps contain similar building blocks as those in SI sampling. Our approach reduces to the standard SI results in the case of a graph representing the conventional time axis whose graph Fourier basis is the discrete Fourier transform (DFT).
We consider two priors on graph signals:
\begin{enumerate}
  \item \textit{Subspace prior}, where the signal lies in a known subspace characterized by a given generator; and
  \item \textit{Smoothness prior}, where the signal is smooth on a given graph.
\end{enumerate}
Both priors parallel those considered in SI sampling \cite{Eldar2009, Eldar2015}.

For the subspace prior, we define the \textit{periodic graph spectrum} (PGS) subspace that serves as a counterpart of SI subspaces. In particular, this subspace maintains the repeated graph frequency spectra of SI signals. 
In the smoothness prior, we assume that the quadratic form of the graph signal is small for a given smoothness function. Under this setting, perfect recovery is no longer possible. Nonetheless, following the work on general Hilbert space sampling, we show how to design graph filters that allow to best approximate the input signal under several different criteria \cite{Eldar2005, Hiraba2007, Dvorki2009, Eldar2015}.

Generalized sampling for standard and graph sampling paradigms allows for the use of arbitrary sampling and reconstruction filters that are not necessarily ideal low-pass filters. It also allows for fixed recovery filters that may have advantages in terms of implementation. For all settings, and under all recovery criteria considered, we show that reconstruction is given by the spectral graph filters, the response of which has a closed form solution that depends on the generator function, smoothness, and sampling/reconstruction filters.

In the context of subspace sampling with a PGS prior, our results allow for a perfect recovery of graph signals beyond those that are bandlimited for almost all signal and sampling spaces. In particular, we require such subspaces to satisfy a direct-sum (DS) condition, as found in standard generalized sampling. When the DS condition does not hold, we design a correction filter that best approximates the input under both the least-squares (LS) and minimax (MX) criteria. These criteria have been studied in the context of standard sampling. We then introduce LS and MX strategies for recovery under a smoothness prior. In all cases, the graph filters have explicit graph frequency responses that parallel those in the SI setting.

Earlier work focusing on generalized sampling of graph signals can be found in \cite{Chepur2018}. This approach is based on generalized Hilbert space sampling \cite{Eldar2009, Eldar2015} and demonstrates the possibility of perfect recovery of graph signals that are not necessarily bandlimited. However, its framework does not parallel SI sampling in general, i.e., the reconstruction matrix does not have a simple diagonal graph frequency response. Likewise, most previous studies have considered vertex domain subsampling, including many graph sampling studies \cite{Anis2016, Chen2015, Marque2016, Chepur2018, Puy2018, Tsitsv2016, Valses2018, Sakiya2019a}, resulting in different building blocks in the sampling and reconstruction steps. Our framework, by contrast, leads to simple closed form recovery methods based on graph filters for both the sampling and recovery. We expand on the similarities and differences between our study and previous approaches in Section \ref{sec:relationship}.

In our preliminary study \cite{Tanaka2019b}, we considered generalized graph sampling with a subspace prior. In this study, the results are significantly expanded by introducing an integrated framework, applying different design criteria, and further considering a smoothness prior.

The remainder of this paper is organized as follows. The notations and basics of GSP are introduced in Section \ref{sec:SGT}. Section \ref{sec:gensamp} reviews generalized sampling in Hilbert spaces and in the SI setting. A framework for generalized graph sampling is presented in Section  \ref{sec:framework}. Section \ref{sec:subspace} proposes signal recovery methods assuming a PGS subspace prior. We describe the relaxation this prior into a smoothness prior in Section \ref{sec:smoothness}. Section \ref{sec:relationship} describes the relationships between our work and existing methods. Numerical experiments are presented in Section \ref{sec:exp}. Finally, Section \ref{sec:conclusion} provides some concluding remarks regarding this research.

\section{Spectral Graph Theory and Basics of GSP}\label{sec:SGT}
We begin by reviewing graphs and their spectrum. We also introduce some basic GSP operators, like the graph Fourier transform (GFT) and filtering on graphs.

A graph $\mathcal{G}$ is represented as $\mathcal{G} = (\mathcal{V}, \mathcal{E})$, where $\mathcal{V}$ and $\mathcal{E}$ denote sets of vertices and edges, respectively. The number of vertices is given by $N=|\mathcal{V}|$ unless otherwise specified. 
We define an adjacency matrix $\mathbf{A}$ with elements
$a_{mn}$ that represents the weight of the edge between the $m$th and $n$th vertices; $a_{mn} = 0$ for unconnected vertices. The degree matrix $\bD$ is a diagonal matrix, with $m$th diagonal element  $[\bD]_{mm} = \sum_n a_{mn}$.

GSP uses different variation operators \cite{Shuman2013, Ortega2018} depending on the application and assumed signal and/or network models. Here, for concreteness, we use the graph Laplacian $\bL:=\bD-\bA$ or its symmetrically normalized version $\underline{\bL} := \bD^{-1/2} \bL \bD^{-1/2}$. The extension to other variation operators (e.g., adjacency matrix) is possible with a slight modification for properly ordering its eigenvalues as long as the graph is undirected without self-loops.
Because $\mathbf{L}$ is a real symmetric matrix, it always possesses an eigendecomposition $\mathbf{L} = \mathbf{U} \bm{\Lambda} \mathbf{U}^*$, where $\mathbf{U} = [{\bm u}_{0}, \ldots, {\bm u}_{N-1}]$ is a unitary matrix containing the eigenvectors $\bm{u}_i$, and $\bm{\Lambda} = \text{diag}(\la_0, \la_1, \ldots, \la_{N-1})$ consists of the eigenvalues $\la_i$. We refer to $\la_i$ as the \textit{graph frequency}.

A graph signal $x: \mathcal{V} \rightarrow \mathbb{C}$ is a signal that assigns a value to each vertex. It can be written as a vector $\boldx$ in which the $n$th
element $x[n]$ represents the signal value at the $n$th vertex.
The GFT is defined as
\begin{equation}
\label{ }
\hat{x}[i] = \langle{\bm u}_{i}, \bm{x}\rangle = \sum_{n=0}^{N-1}u^*_{i}[n]x[n].
\end{equation}
Our generalized sampling can also use other GFT definitions, e.g., \cite{Deri2017, Giraul2018, LeMag2018, Lu2019}, without changing the framework.

A (linear) graph filter is defined as $\bG \in \mathbb{C}^{N\times N}$. The filtered signal is represented as
\begin{equation}
\label{eqn:graphfilter}
\boldx_{\text{out}} = \mathbf{G} \boldx.
\end{equation}
Graph filtering may be defined in the vertex and frequency domains. Vertex domain filtering is defined as a linear combination of the neighborhood samples
\begin{equation}
\label{eqn:vertex_filtering}
x_{\text{out}}[n] := \sum_{i \in \mathcal{N}_n} [\mathbf{G}]_{ni} \ x[i],
\end{equation}
where $\mathcal{N}_n$ represents neighborhood vertex indices around the $n$th vertex.
In graph frequency domain filtering, the output is defined as a generalized convolution \cite{Shuman2016a}:
\begin{equation}
\label{eqn:GFT_filtering}
x_{\text{out}}[n] := \sum_{i = 0}^{N-1} \hat{x}[i] G(\li) u_i[n]
\end{equation}
where the filter response in the graph frequency domain is given by $G(\li) \in \mathbb{R}$. This filtering is equivalently written as
\begin{equation}
\label{eqn:GFT_filtering_matrix}
\boldx_{\text{out}} = \bU G(\bLam) \bU^* \boldx,
\end{equation}
where $G(\bLam):=\text{diag}(G(\la_0), G(\la_1), \dots)$. Here, $\bG = \bU G(\bLam) \bU^*$. If $G(\li)$ is a $P$th order polynomial, \eqref{eqn:GFT_filtering} coincides with vertex domain filtering \eqref{eqn:vertex_filtering} with a $P$-hop local neighborhood \cite{Shuman2013}.

\section{Generalized Sampling in Hilbert Space}\label{sec:gensamp}
This section introduces prior results on generalized sampling in Hilbert spaces \cite{Eldar2003, Eldar2006, Eldar2015} and corresponding results in the SI setting, which are fundamental for our generalized graph sampling approach. This section briefly describes the generalized sampling framework.
Detailed derivations, including error analysis which can be easily applied to graph sampling, may be found in \cite{Eldar2015}.
Table \ref{tb:SI_filters} summarizes the main results of this section in the SI setting.

\begin{table*}[t]
\caption{Correction and Reconstruction Filters for Shift-Invariant and Graph Spectral Filters where CF and RF are abbreviations of correction filter and reconstruction filter, respectively. DS, LS, and MX refer to direct-sum, least squares, and minimax solutions, respectively. Spectra $R_{XY}(\omega)$ and $\tilde{R}_{XY}(\la_i)$ are defined in \eqref{eqn:sampled_cc} and \eqref{eqn:graph_cc}.}\label{tb:SI_filters}
\centering
\renewcommand{\arraystretch}{1}
\begin{tabular}{c||c|c|cl||c|c|cl}
\hline
& \multicolumn{4}{c||}{Shift-invariant subspace}& \multicolumn{4}{c}{Periodic graph spectrum subspace} \\\cline{2-9}
& \multicolumn{2}{c|}{Unconstrained}&  \multicolumn{2}{c||}{Predefined ($W(\omega)$ is fixed)}&  \multicolumn{2}{c|}{Unconstrained}&  \multicolumn{2}{c}{Predefined ($W(\li)$ is fixed)} \\\hline
Filter& CF& RF& \multicolumn{2}{c||}{CF} & CF& RF& \multicolumn{2}{c}{CF} \\\hline
\multirow{3}{*}{Subspace Prior}& \multirow{2}{*}{$\dfrac{1}{R_{SA}(\omega)}$} & \multirow{2}{*}{$A(\omega)$} & $\dfrac{R_{WA}(\omega)}{R_{SA}(\omega)R_{WW}(\omega)}$ & DS, MX& \multirow{2}{*}{$\dfrac{1}{\tilde{R}_{SA}(\la_i)}$}&\multirow{2}{*}{$A(\la_i)$} & $\dfrac{\tilde{R}_{WA}(\la_i)}{\tilde{R}_{SA}(\la_i)\tilde{R}_{WW}(\la_i)}$ & DS, MX\\
& & & $\dfrac{1}{R_{SW}(\omega)}$ & LS& & &  $\dfrac{1}{\tilde{R}_{SW}(\la_i)}$& LS\\\hline
\multirow{4}{*}{Smoothness Prior}&\multirow{2}{*}{$\dfrac{1}{R_{SW}(\omega)}$} & \multirow{2}{*}{$\dfrac{S(\omega)}{|V(\omega)|^2}$} & $\dfrac{1}{R_{S\widetilde{W}}(\omega)}$& LS& \multirow{2}{*}{$\dfrac{1}{\tilde{R}_{SW}(\la_i)}$}& \multirow{2}{*}{$\dfrac{S(\la_i)}{V^2(\la_i)}$}& $\dfrac{1}{\tilde{R}_{S\widetilde{W}}(\la_i)}$& LS\\
& & &$\dfrac{R_{W\widetilde{W}}(\omega)}{R_{S\widetilde{W}}(\omega)R_{WW}(\omega)}$ & MX& & & $\dfrac{\tilde{R}_{W\widetilde{W}}(\li)}{\tilde{R}_{S\widetilde{W}}(\li)\tilde{R}_{WW}(\li)}$& MX\\\hline
\end{tabular}
\end{table*}%

\subsection{Sampling and Recovery Framework}
Fig. \ref{fig:generalized_sampling_classical}(a) illustrates the generalized sampling framework in Hilbert space. Its SI counterpart is shown in Fig. \ref{fig:generalized_sampling_classical}(b). Let $x$ be a vector in a Hilbert space $\mH$ and $c[n]$ be its $n$th sample given by $c[n] = \langle s_n, x\rangle$, where $\{s_n\}$ is a Riesz basis and $\langle \cdot, \cdot \rangle$ is an inner product. Denoting by $S$ the set transformation corresponding to $\{s_n\}$ we can write the samples as $c = S^* x$, where $\cdot^*$ represents the adjoint. The subspace generated by $\{s_n\}$ is denoted by $\mS$.

In the SI setting, $s_n = s(t-nT)$ for a real function $s(t)$ and a given period $T$. The samples can then be expressed as
\begin{equation}
\label{eqn:cn}
c[n] = \langle s(t-nT), x(t)\rangle = \left. x(t) \ast s(-t)\right|_{t = nT},
\end{equation}
where $\ast$ represents convolution. The continuous-time Fourier transform (CTFT) of the samples $c[n]$, denoted $C(\omega)$, can be written as
\begin{equation}
\label{eqn:Comega}
C(\omega) = R_{SX}(\omega),
\end{equation}
where
\begin{equation}
\label{eqn:sampled_cc}
R_{SX}(\omega):=\frac{1}{T}\sum_{k=-\infty}^{\infty}S^*\left(\frac{\omega - 2\pi k}{T}\right)X\left(\frac{\omega - 2\pi k}{T}\right)
\end{equation}
is the sampled cross correlation. Thus, we may view sampling in the Fourier domain as multiplying the input spectrum by the filter's frequency response and subsequently aliasing the result with uniform intervals that depend on the sampling period. In bandlimited sampling, $s(-t) = \text{sinc}(t/T)$, where $\text{sinc}(t) = \sin(\pi t)/(\pi t)$. However, $s(t)$ can be chosen arbitrarily in the generalized sampling framework.

The recovery of the sampled signal $c$ is represented as
\begin{equation}
\label{eqn:xtilde_hilbert}
\tilde{x} = WHc = WH(S^*x),
\end{equation}
where $W$ is a set transformation corresponding to a basis $\{w_n\}$ for the reconstruction space, which spans a closed subspace $\mW$ of $\mH$. The transform $H$ is called the \emph{correction transformation} and operates on the samples $c$ prior to recovery.

In the SI setting, the recovery corresponding to \eqref{eqn:xtilde_hilbert} is given by
\begin{equation}
\label{eqn:x_tilde}
\tilde{x}(t) = \sum_{n \in \mathbb{Z}} (h[n] \ast c[n]) w(t-nT),
\end{equation}
where a discrete-time correction filter $h[n]$ is first applied to $c[n]$: The output $d[n] = h[n] \ast c[n]$ is interpolated by $w(t-nT)$, to produce the recovery $\tilde{x}(t)$.

Next, we describe known results on generalized sampling with subspace and smoothness priors.

\begin{figure}[t]
\centering
\subfigure[][Sampling in Hilbert space]{\includegraphics[width = \linewidth]{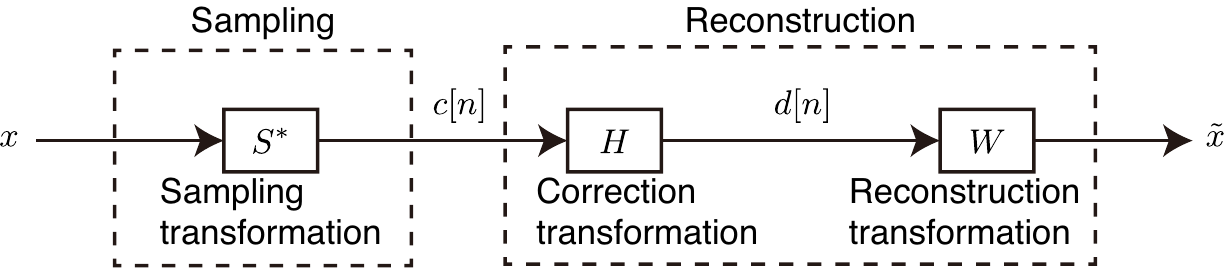}}\\
\subfigure[][Sampling in SI space]{\includegraphics[width = \linewidth]{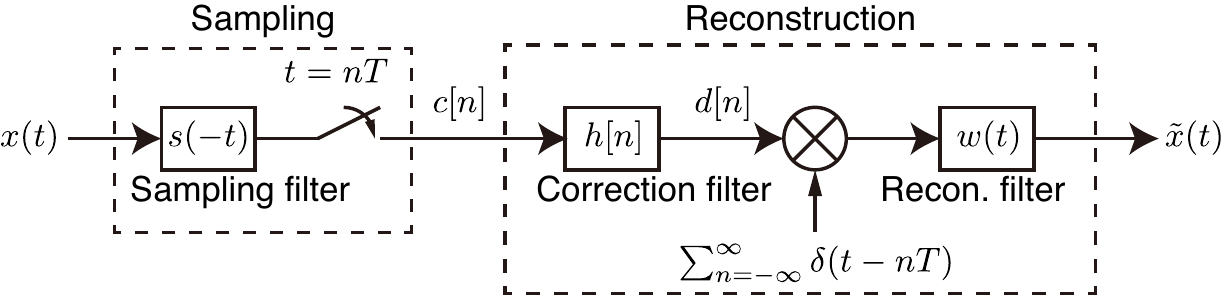}}
\caption{Generalized sampling frameworks for sampling in Hilbert and SI spaces. The same sampling-correction-reconstruction system can be used for both the subspace and smoothness priors. In Hilbert space sampling, the original and reconstructed signals are represented as $x$ and $\tilde{x}$, respectively. The sequence of samples is $\{c[n]\}$ and its corrected counterpart is $\{d[n]\}$. The set transformations of sampling and reconstruction are given by $S: \mH \rightarrow \ell_2$ and $W: \ell_2 \rightarrow \mH$, respectively. The corrected sequence is given by $d = Hc$ for some linear correction transformation $H: \ell_2 \rightarrow \ell_2$. For SI sampling, the original and reconstructed continuous-time signals are represented as $x(t)$ and $\tilde{x}(t)$, respectively. The set transformations reduce to filtering, as indicated in the figure.}
\label{fig:generalized_sampling_classical}
\end{figure}

\subsection{Subspace Prior}\label{subsec:ss_hilbert}
Suppose that $x$ lies in an arbitrary subspace $\mA$ of $\mH$ and assume that $\mA$ is known. Hence, $x$ can be represented as $x = \sum d[n] a_n = Ad$, where $\{a_n\}$ is an orthonormal basis for $\mA$ and $d[n]$ are the expansion coefficients of $x$.
In the SI setting, $x(t)$ is expressed as
\begin{equation}
\label{eqn:xt_si}
x(t) = \sum_{n \in \mathbb{Z}} d[n]a(t-nT),
\end{equation}
for some sequence $d[n]$ where $a(t)$ is a real generator satisfying the Riesz condition.
In the Fourier domain, \eqref{eqn:xt_si} becomes 
\begin{equation}
\label{eq:SI}
X(\omega)= D(e^{j \omega T}) A(\omega),
\end{equation}
where $A(\omega)$ is the CTFT of $a(t)$ and $D(e^{j \omega T})$ 
is the discrete-time Fourier transform (DTFT) of the sequence $d[n]$, and is $2\pi/T$ periodic.

\subsubsection{Unconstrained Case}
We first consider the case in which the recovery is unconstrained, so that $W$ can be any transformation. In this setting, we may recover a signal in $\mathcal{A}$ by choosing $W=A$ in \eqref{eqn:xtilde_hilbert}. If $S^*A$ is invertible, then perfect recovery of any $x \in \mathcal{A}$ is possible by using
$H=(S^*A)^{-1}$. Invertibility can be ensured by the DS condition: $\mA$ and $\mS^\bot$ intersect only at the origin and span $\mH$ jointly. This requirement is formally written as
\begin{equation}
\label{eqn:DS_condition}
\mH = \mA \oplus \mS^\bot.
\end{equation}

Under the DS condition, a unique recovery is obtained by an oblique projection operator onto $\mA$ along $\mS^\bot$ given by
\begin{equation}
\label{eqn:x_ss_ds_h}
\tilde{x} = A (S^*A)^{-1} S^* x = x.
\end{equation}
In the SI setting, the frequency response of the correction filter is
\begin{equation}
\label{eqn:Homega}
H(\omega) =  \frac{1}{R_{SA}(\omega)}.
\end{equation}

If $\mA$ and $\mS^\bot$ intersect, then there is more than one signal in $\mA$ that matches the sampled signal $c$. We may then consider several selection criteria to obtain an appropriate signal out of (infinitely) many candidates. Widely accepted strategies are the LS and MX approaches.

The LS recovery is the minimum energy solution obtained as
\begin{equation}
\label{ }
\tilde{x} = \argmin_{x\in\mA,\ S^*x = c} \|x\|^2,
\end{equation}
and is given by
\begin{equation}
\label{eqn:x_ss_ls_h}
\tilde{x} = A (S^*A)^{\dagger} S^* x.
\end{equation}
Here, $H = (S^*A)^{\dagger}$ and $\cdot^{\dagger}$ represents the Moore-Penrose pseudo inverse.
Its corresponding form in the SI setting is
\begin{equation}
\label{ }
H(\omega) = \begin{cases}
\frac{1}{R_{SA}(\omega)} & R_{SA}(\omega) \neq 0\\
0 & R_{SA}(\omega) = 0.
\end{cases}
\end{equation}

The MX criterion minimizes the worst-case error from the original signal:
\begin{equation}
\label{ }
\tilde{x} = \argmin_{\tilde{x}} \max_{x \in \mA,\ S^*x = c} \|\tilde{x} - x\|^2.
\end{equation}
The solution with a subspace prior is the same as that in \eqref{eqn:x_ss_ls_h}.

\subsubsection{Predefined Case}
When the reconstruction transformation $W$ is predefined, perfect recovery is not possible in general. However, we can still design a correction transformation $H$ such that the solution is close to $x$ in some sense.

With the DS condition in \eqref{eqn:DS_condition}, a minimal-error recovery can be obtained by the correction filter
\begin{equation}
\label{eqn:H_ss_ds_h}
H = (W^*W)^{-1} W^* A (S^*A)^{-1}.
\end{equation}
The recovered signal is $\tilde{x} = W(W^*W)^{-1} W^* A (S^*A)^{-1}S^* x$, which is the orthogonal projection of the unconstrained solution onto $\mW$. In the SI setting,
\begin{equation}
\label{eqn:SI_H_ss_ds}
H(\omega) = \frac{R_{WA}(\omega)}{R_{SA}(\omega)R_{WW}(\omega)}.
\end{equation}

When the DS condition does not hold, the LS and MX strategies can be considered as in the unconstrained case.
The LS solution is $H = (S^*W)^{\dagger}$, which results in the following reconstruction:
\begin{equation}
\label{eqn:x_ss_pd_ls}
\tilde{x} = W(S^*W)^{\dagger} S^* x.
\end{equation}
This solution is the same as that in \eqref{eqn:x_ss_ls_h} by replacing $A$ with $W$.
The MX solution is given by
\begin{equation}
\label{ }
\tilde{x} = W(W^*W)^{-1} W^* A(S^* A)^{\dagger} S^* x,
\end{equation}
with $H = (W^*W)^{-1} W^* A(S^* A)^{\dagger}$. The corresponding SI solution is the same as that in \eqref{eqn:SI_H_ss_ds}
with $H(\omega) =0$ when the denominator is zero.

\subsection{Smoothness Prior}\label{subsec:sm_hilbert}
The smoothness prior is a less restrictive assumption than the subspace prior because the actual signal subspace $\mA$ is not necessarily known. Instead, we assume the signal is smooth, which is formulated as $\|Vx\| \leq \rho$ for some invertible operator $V$: In the SI setting $V = V(\omega)$ is nonzero for all $\omega$. Smoothness is often measured by low energy in high frequency components:
\begin{equation}
\label{eqn:high_freq_energy_SI}
\int_{-\infty}^{\infty} |V(\omega) X(\omega)|^2 d\omega \leq \rho^2.
\end{equation}

In general, with a smoothness prior, there are infinitely many solutions. Two approaches to select a solution are the LS and MX methods, which can be applied in both the unconstrained and constrained settings.

\subsubsection{Unconstrained Case}
Suppose that $V^*V$ is a bounded operator. In the LS method, the objective function is formulated by choosing the smoothest signal among all the possibilities:
\begin{equation}
\label{eqn:x_sm_costfunc}
\tilde{x} = \argmin_{x \in \{x| S^*x = c\}} \|Vx\|^2.
\end{equation}
The solution to \eqref{eqn:x_sm_costfunc} is given by
\begin{equation}
\label{eqn:x_sm_unc}
\tilde{x} = \widetilde{W} (S^* \widetilde{W})^{-1} S^* x
\end{equation}
where $\widetilde{W} = (V^*V)^{-1} S$. In the SI setting, the correction filter in \eqref{eqn:x_sm_unc} reduces to
\begin{equation}
\label{ }
H(\omega)  = \frac{1}{R_{S\widetilde{W}}(\omega)}
\end{equation}
with
\begin{equation}
\label{eqn:W_smoothness_unc}
\widetilde{W}(\omega) = \frac{S(\omega)}{|V^2(\omega)|}.
\end{equation}
The MX solution coincides with \eqref{eqn:x_sm_unc}.

\subsubsection{Predefined Case}
When the recovery space is predefined, the constraint on the feasible set is slightly different from that in \eqref{eqn:x_sm_costfunc}.
The LS objective for the predefined case may be formulated as
\begin{equation}
\label{eqn:sm_costfunc_pd}
\tilde{x} = \argmin_{x \in \{x| x\in \mW,\  S^*x = Pc\}} \|Vx\|^2,
\end{equation}
where $P$ is the orthogonal projection onto the range space of $S^*W$.
The solution can be shown to be given by
\begin{equation}
\label{ }
\tilde{x} = \widehat{W} (S^* \widehat{W})^\dagger S^* x
\end{equation}
where $\widehat{W} = W(W^* V^* V W)^{-1} W^* S$. In the SI setting, \eqref{eqn:sm_costfunc_pd} reduces to the use of $H(\omega) = 1/R_{SW}(\omega)$ prior to reconstruction with $W(\omega)$ \cite[Section 7.2.1]{Eldar2015}. Therefore, constrained recovery under the LS objective is the same in the subspace and smoothness priors and the smoothness constraint is not included in the solution.

The MX criterion with a smoothness prior can be formulated as
\begin{equation}
\label{ }
\tilde{x} = \argmin_{\tilde{x} \in \mW} \max_{x \in \{x| S^*x = c,\ \|Vx\| \leq \rho\}} \|\tilde{x} - W(W^*W)^{-1}Wx\|^2.
\end{equation}
This solution is given by
\begin{equation}
\label{eqn:x_sm_pd_mx}
\tilde{x} = W(W^*W)^{-1}W \widetilde{W} (S^* \widetilde{W})^{-1} S^* x.
\end{equation}
This is the orthogonal projection onto $\mW$ of the unconstrained solution in \eqref{eqn:x_sm_unc}: The correction transformation is $H = (W^*W)^{-1}W \widetilde{W} (S^* \widetilde{W})^{-1}$. In the SI setting, it reduces to
\begin{equation}
\label{ }
H(\omega) = \frac{R_{W\widetilde{W}}(\omega)}{R_{S\widetilde{W}}(\omega)R_{WW}(\omega)}.
\end{equation}

\section{Sampling and Recovery of Graph Signals}\label{sec:framework}
\subsection{Sampling of Graph Signals}\label{subsec:graphsampling}
Two methods of sampling over graphs have been proposed in the literature: 1) sampling in the vertex domain \cite{Anis2016, Chen2015} and 2) sampling in the graph frequency domain \cite{Tanaka2018}.

\subsubsection{Sampling in the Vertex Domain}
For sampling in the vertex domain, samples on a predetermined vertex set $\mT$ are selected. This corresponds to nonuniform subsampling in the time domain. In contrast to the SI setting, vertex domain sampling is conducted nonuniformly because vertex indices do not reflect the structure of the signal. Many approaches have been proposed to select the ``best'' sampling set from a given graph under different criteria \cite{Anis2016, Puy2018, Sakiya2019a, Tsitsv2016}.

Let us define $\bI_\mT \in \{0, 1\}^{K\times N}$ as a submatrix of the identity matrix $\bI_N$, whose rows are determined by the sampling set $\mT$ that identifies the vertices that remain after sampling, i.e., row indices in $\bI_N$. Sampling in the vertex domain is defined as follows:

\begin{definition}[Sampling of graph signals in the vertex domain \cite{Anis2016, Chen2015}]\label{def:GD_vertex}
Let $\boldx \in \mathbb{C}^N$ be the original graph signal and $\bG \in \mathbb{C}^{N \times N}$ be an arbitrary graph filter in \eqref{eqn:graphfilter}. In addition, let $\bI_\mT$ be a submatrix of the identity matrix $\bI_N$ extracting $K = |\mT|$ rows corresponding to the sampling set $\mT$. The sampled graph signal $\bmc \in \mathbb{C}^K$ is given as follows:
\begin{equation}
\label{eqn:vertex_samp}
\bmc = \bI_\mT \mathbf{G}\boldx.
\end{equation}
\end{definition}
\noindent
The sampling matrix is therefore given by $\bS^* = \bI_\mT \mathbf{G}$.

Aggregation sampling \cite{Marque2016, Valses2018} is a variant of vertex sampling that uses a specifically designed $\mathbf{G}$. For example, \cite{Marque2016} defines
\begin{equation}
\label{eqn:G_agg}
\mathbf{G} = \bm{\Psi}\text{diag}(u_0^*(\li), u_1^*(\li), \dots) \bU^*
\end{equation}
where $[\bm{\Psi}]_{k \ell} =  \la_\ell^{k}$, and \cite{Valses2018} utilizes a random matrix to filter the signal, i.e.,
\begin{equation}
\label{eqn:G_agg_rand}
\mathbf{G} = (\bI + \bA) \circ \mathbf{\Xi},
\end{equation}
where $\mathbf{\Xi}$ is a random matrix and $\circ$ represents an element-wise product.
In general, $\mathbf{G}$ in \eqref{eqn:G_agg} and \eqref{eqn:G_agg_rand} cannot be decomposed using $\bU$ and therefore it does not have a diagonal graph frequency response.

The definitions above based on vertex domain operations result in nonuniform sampling in general. When the signal is bandlimited (in a graph frequency sense), perfect recovery is guaranteed if $\mT$ is a \textit{uniqueness set} \cite{Anis2016, Pesens2008}. However, sampling and reconstruction are not symmetric in general: Recovery is not given through filtering and upsampling. This is a significant difference from SI sampling.
Instead, we use frequency domain sampling to build a parallel of generalized SI sampling applied to the graph setting and enable recovery through frequency domain filtering. 

\subsubsection{Sampling in the Graph Frequency Domain}
To define sampling over a graph, we extend sampling in SI subspaces expressed by \eqref{eqn:Comega} to the graph frequency domain \cite{Tanaka2018}.
In particular, the graph Fourier transformed input $\hat{\boldx}$ is first multiplied by a graph frequency filter $S(\bLam)$; the product is subsequently aliased with period $K$. This results in the following definition:

\begin{definition}[Sampling of graph signals within the graph frequency domain]\label{def:GD_spectral}
Let $\hat{\boldx} \in \mathbb{C}^N$ be the original signal in the graph frequency domain, i.e., $\hat{\boldx} = \bU^* \boldx$, and let $S(\bLam)$ be an arbitrary sampling filter in the graph frequency domain. For any sampling ratio $M\in \mathbb{Z}$, the sampled graph signal in the graph frequency domain\footnote{$M$ is assumed to be a divisor of $N$ for simplicity.} is given by $\hat{\bmc} \in \mathbb{C}^{K}$, where $K=N/M$, and 
\begin{equation}
\label{eqn:gft_sampling}
\hat{c}(\la_i) = \sum_{\ell=0}^{M-1} S\left(\la_{i + K\ell}\right)\hat{x}\left(\la_{i + K\ell}\right).
\end{equation}
In matrix form, the sampled graph signal can be represented as $\hat{\bmc} = \bD_{\text{\emph{samp}}} S(\bLam)\hat{\boldx}$
where $\bD_{\text{\emph{samp}}} = \begin{bmatrix} \mathbf{I}_{K} & \mathbf{I}_{K} & \ldots \end{bmatrix}$.
\end{definition}
\noindent
Hereafter, we denote the sampling matrix $\bS^*$ as follows.
\begin{equation}
\label{eqn:S}
\bS^* = \bD_{\text{samp}} S(\bLam) \bU^*.
\end{equation}

This graph frequency domain sampling ``mixes'' different frequency components obtained by the GFT. Different eigenvectors represent different variations on the graph \cite{Shuman2013, Shuman2016a, Perrau2018}. The GFT coefficients of a graph signal provide a notion of a frequency content similar to the DFT; however, the GFT basis varies according to the underlying graph and variation operator used. The (weighted) sum of the two GFT coefficients in \eqref{eqn:gft_sampling} can be seen as a counterpart of sampling in the Fourier domain in classical signal processing.

Suppose that $\bU^*$ is the DFT matrix: For example, the DFT matrix diagonalizes the graph Laplacian $\bL$ of the circular graph \cite{Strang1999}. In this case, the GFT domain sampling in \eqref{eqn:gft_sampling} coincides with that in the DFT domain \cite{Tanaka2018}, i.e., the sampled spectrum of \eqref{eqn:Comega} $C[i] = C(2\pi i/N)$ ($i = 0, \dots, N-1$) yields the same output as in \eqref{eqn:gft_sampling}.

\begin{figure}[t]
   \centering
   \includegraphics[width = .7\linewidth]{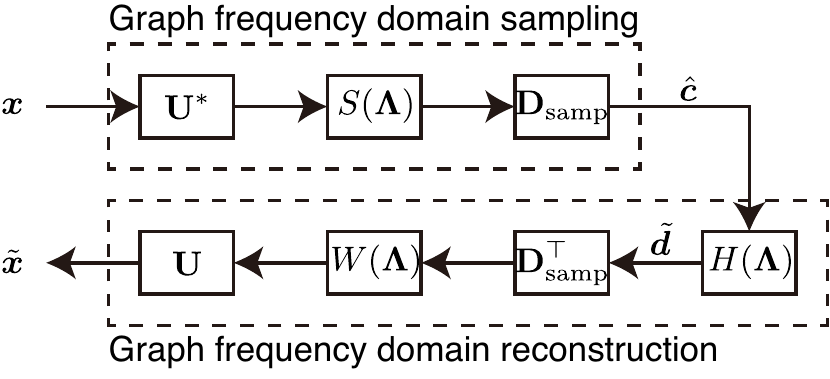}
   \caption{Generalized sampling framework for PGS subspaces. Here, $\boldx$ and $\tilde{\boldx}$ are the original and reconstructed graph signals, respectively, $\hat{\bmc}$ is the sampled signal in the graph frequency domain, and $\tilde{\bmd}$ is the corrected graph signal.}
   \label{fig:generalized_sampling}
\end{figure}

\subsection{Recovery Framework for Generalized Graph Sampling}
Our framework for generalized graph sampling is shown in Fig. \ref{fig:generalized_sampling}. It parallels sampling in Hilbert and SI spaces, as illustrated in Fig.~\ref{fig:generalized_sampling_classical} \cite{Eldar2006, Eldar2009}. In this paper, we assume that sampling, filtering, and reconstruction are all performed in the graph frequency domain. This results in graph filters that can be interpreted as an analog of SI sampling. As in standard sampling theory, three filters are critical in the recovery problem: sampling, correction, and reconstruction filters.

To sample $\boldx$, we transform the input into the GFT domain, resulting in $\hat{\boldx} = \bU^* \boldx$. The output is subsequently filtered using the sampling filter $S(\bLam)$.
The filtered signal is downsampled to yield a sampled signal $\hat{\bmc} =\bS^* \boldx = \bD_{\text{samp}}S(\bLam)\hat{\boldx}$. In the reconstruction step, $\hat{\bmc}$ is filtered by the correction filter $\bH = H(\bLam):= \text{diag}(H(\la_0), \dots, H(\la_{N-1}))$.
Finally, $\tilde{\bmd} = H(\bLam) \hat{\bmc}$ is upsampled to the original dimension by $\bD_{\text{samp}}^\top$, and the reconstruction filter $W(\bLam) := \text{diag}(W(\la_0), \dots, W(\la_{N-1}))$ is applied to the upsampled signal. After performing an inverse GFT, we obtain the recovered signal $\tilde{\boldx}$. This can be written as
$\tilde{\boldx} = \bU W(\bLam) \bD_{\text{samp}}^\top \bH \hat{\bmc}$, where the reconstruction matrix is given by $\bW := \bU W(\bLam) \bD_{\text{samp}}^\top$.

The primary objective in this framework is to consider the design method of the correction and reconstruction filters, $\bH$ and $\bW$, that recover the original signal as accurately as possible with a given prior and constraint. We follow the same strategies as that of generalized sampling in Hilbert spaces introduced in Section \ref{sec:gensamp}: DS, LS, and MX. The solutions with subspace and smoothness priors are presented in Sections \ref{subsec:ss_hilbert} and \ref{subsec:sm_hilbert}, respectively.

\begin{figure}[tp]
   \centering
   \includegraphics[width = \linewidth]{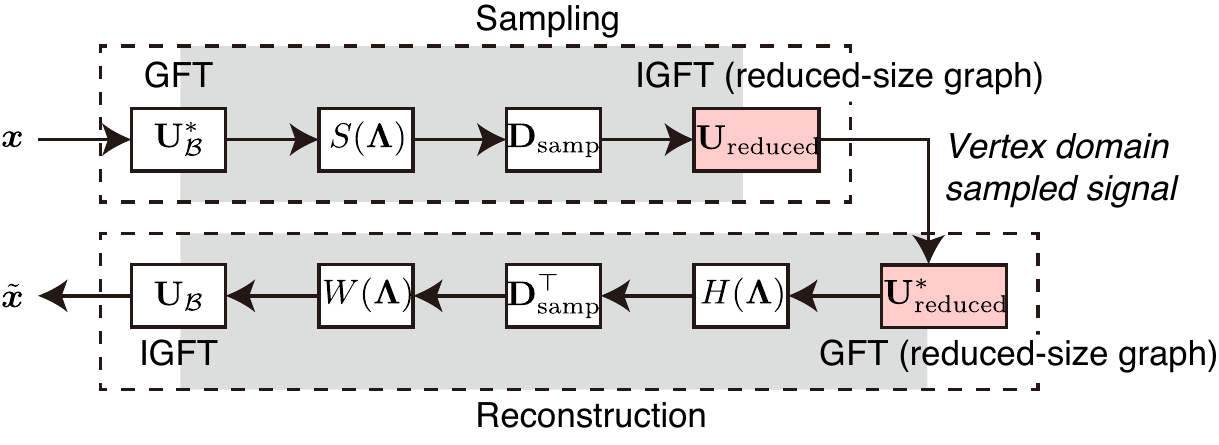} 
   \caption{Generalized sampling framework for graph signals using frequency domain sampling; the sampled signal is transformed back into the vertex domain. The red boxes are building blocks from Fig. \ref{fig:generalized_sampling}. The sampled signal is located on the vertices of $\mathcal{G}_{\text{reduced}}$. The gray regions correspond to the graph frequency domain.}
   \label{fig:gsp_sampling_framework_vertex}
\end{figure}

\begin{figure*}[tp]
   \centering
   \includegraphics[width = \linewidth]{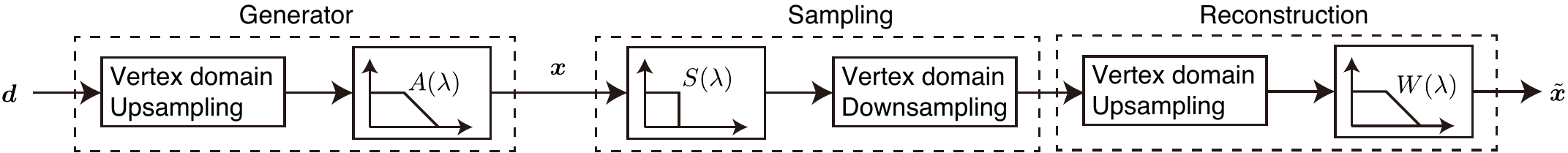} 
   \caption{Generalized graph sampling framework for a bipartite graph. A non-bandlimited graph signal $\boldx$ is generated from the expansion coefficients $\bmd$ by performing the vertex domain upsampling and filtering by $A(\la)$. Subsequently, $\boldx$ is filtered with the sampling filter $S(\la)$, followed by vertex domain subsampling. The signal is reconstructed by applying the vertex domain upsampling followed by $W(\la)$ in \eqref{eqn:W_symnormalized}.}
   \label{fig:bpt_example}
\end{figure*}

\subsection{Sampling and Recovery on Bipartite Graphs}\label{subsec:bpt}
Before describing the filter design methods in the following sections, we introduce an interesting and special case, namely, sampling on bipartite graphs. In this scenario, vertex and spectral domain sampling coincide.

Studies on sampling for bipartite graphs are not only interesting as an interconnection between two sampling paradigms, but also of practical importance. Bipartite graph sampling has been well studied in the context of perfect reconstruction graph filter banks \cite{Narang2012, Narang2013, Sakiya2014a, Sakiya2016a, Tanaka2014a, Tanaka2019a}, where the original signal is decomposed through low- and high-pass channels. Maximally decimated coefficients are obtained by applying sampling for each channel where the two disjoint vertex sets correspond to the transformed coefficients in each channel. To realize perfect recovery for an arbitrary graph, one needs to decompose the original (non-bipartite) graph into several bipartite graphs with disjoint edge subsets\footnote{An arbitrary $\mathcal{X}$-colorable graph can always be decomposed into $\lceil \log_2 \mathcal{X} \rceil$ bipartite subgraphs with disjoint edge subsets \cite{Narang2012, Harary1977}.}. Multiband decompositions are then recursively conducted on these bipartite graphs. Such a filter bank has multiple channels to realize recovery of full-band graph signals. Instead, our graph sampling framework only has one branch, as shown in Fig. \ref{fig:generalized_sampling}, while still allowing for perfect recovery under appropriate conditions.

Suppose that a graph is bipartite having two equal-sized vertex sets. Formally, let $\mathcal{B} = (\mathcal{V}_1, \mathcal{V}_2, \mathcal{E})$ be a bipartite graph that contains two disjoint vertex sets $\mathcal{V}_1$ and $\mathcal{V}_2$, where $|\mathcal{V}_1| = |\mathcal{V}_2| = N/2$, i.e., edges only exist between $\mathcal{V}_1$ and $\mathcal{V}_2$. Without loss of generality, we assume that its first $N/2$ vertices correspond to $\mathcal{V}_1$ and its last vertices correspond to $\mathcal{V}_2$. We also assume that the GFT matrix is the eigenvector matrix of the symmetric normalized graph Laplacian.

Fig. \ref{fig:gsp_sampling_framework_vertex} illustrates the vertex domain representation of our generalized sampling framework of Fig. \ref{fig:generalized_sampling}, where the sampled signal is transformed back into the vertex domain. Suppose that the reduced-size graph $\mathcal{G}_{\text{reduced}}$ of size $N/2$ is obtained by reconnecting edges within $\mathcal{V}_1$ with Kron reduction \cite{Dorfle2013, Shuman2016}. The symmetric normalized graph Laplacian of $\mathcal{G}_{\text{reduced}}$ can be obtained from $\underline{\bL}$ of $\mB$ through the following:
\begin{equation}
\label{eqn:Kronred}
\underline{\bL}_{\text{reduced}} =\underline{\bL}_{\mV_1 \mV_1}-\underline{\bL}_{\mV_1 \mV_2}\underline{\bL}_{\mV_2\mV_2}^{-1}\underline{\bL}_{\mV_2 \mV_1},
\end{equation}
where $\underline{\bL}_{\mathcal{X} \mathcal{Y}}$ is a submatrix of $\underline{\bL}$, whose extracted rows and columns from $\underline{\bL}$ are specified by $\mathcal{X}$ and $\mathcal{Y}$, respectively.

Considering the simplest direct sampling, i.e., there is no sampling filter. The relationship between sampling in the vertex domain \eqref{eqn:vertex_samp} and the vertex domain representation of sampling in the graph frequency domain \eqref{eqn:gft_sampling} is given in the following theorem, taken from \cite{Sakiya2019}:
\begin{theorem}\label{th:vertex_freq}
Suppose that the GFT matrix $\bU_{\mB}$ is the eigenvector matrix of the symmetric normalized graph Laplacian of $\mathcal{B}$, and $\bU_{\emph{reduced}}$ is the eigenvector matrix of $\underline{\bL}_{\emph{reduced}}$ in \eqref{eqn:Kronred}. The following relationship then holds:
\begin{equation}
\label{eqn:vertex_gft_sampling}
\bU_{\emph{reduced}} \bD_{\emph{samp}}
\bU_{\mB}^* =  \begin{bmatrix}
\bI & \mathbf{0}
\end{bmatrix} = \bI_{\mathcal{V}_1}.
\end{equation}
\end{theorem}
\noindent
This relationship indicates that vertex domain sampling (RHS in \eqref{eqn:vertex_gft_sampling}) coincides with graph frequency domain sampling (LHS in \eqref{eqn:vertex_gft_sampling}) under the special case in which the graph is bipartite and the graph filter is the identity operator, i.e., $\bG = \bI$.

Theorem \ref{th:vertex_freq} can be generalized to allow sampling with a sampling filter $\bG$ not necessarily being the identity:

\begin{corollary}\label{col:sampling_ver_gft}
Suppose the same setup as in Theorem \ref{th:vertex_freq} where a sampling filter $\mathbf{G}$ is conducted before subsampling, as in \eqref{eqn:vertex_samp}. If the sampling filter $\mathbf{G}$ is diagonalizable by $\bU_{\mB}$, then $\bU_{\emph{reduced}}\bS^*$, i.e., the vertex domain representation of graph frequency domain sampling, is identical to $\bI_{\mathcal{V}_1}\mathbf{G}$.
\end{corollary}

\begin{proof}
By assumption, $\mathbf{G}$ can be represented as follows:
\begin{equation}
\label{eqn:G_USU}
\mathbf{G} = \bU_{\mB} S(\bLam) \bU_{\mB}^*.
\end{equation}
Therefore, the sampling matrix $\bS^*$ in \eqref{eqn:S} is equal to
\begin{equation}
\begin{split}
\bS^* & = \bD_{\text{samp}} S(\bLam) \bU_{\mB}^*\\
& = \bD_{\text{samp}} \bU_{\mB}^*\bU_{\mB} S(\bLam)\bU_{\mB}^*\\
&= \bD_{\text{samp}} \bU_{\mB}^* \mathbf{G}.
\end{split}
\end{equation}
Using \eqref{eqn:vertex_gft_sampling},
\begin{equation}
\label{eqn:UredS}
\bU_{\text{reduced}}\bS^*= \bU_{\text{reduced}} \bD_{\text{samp}} \bU_{\mB}^* \mathbf{G} = \bI_{\mathcal{V}_1}\mathbf{G},
\end{equation}
completing the proof.
\end{proof}

Corollary \ref{col:sampling_ver_gft} addresses the sampling phase. Similarly, during reconstruction, the correction-then-reconstruction steps in the graph frequency domain, $\bW \bH \bU^*_{\text{reduced}}$, can be jointly represented using vertex domain upsampling:
\begin{equation}
\label{eqn:W_symnormalized}
\begin{split}
\bW\bH \bU_{\text{reduced}} &= \bU_{\mB} W(\bLam) \bD_{\text{samp}}^\top H(\bLam) \bU^*_{\text{reduced}}\\
& = \bU_{\mB} W(\bLam) \text{diag}(H(\bLam), H(\bLam)) \bD_{\text{samp}}^\top\bU^*_{\text{reduced}}\\
& = \bU_{\mB} W'(\bLam)  \bU^*_{\mB} \bU_{\mB}\bD_{\text{samp}}^\top\bU^*_{\text{reduced}}\\
&= \bU_{\mB} W'(\bLam) \bU^*_{\mB}
\begin{bmatrix}
\bI & \mathbf{0}
\end{bmatrix}^\top\\
& = \bW' \bI_{\mathcal{V}_1}^\top
\end{split}
\end{equation}
where $W'(\li) := W(\li)H(\la_{i\text{ mod }N/2})$.

Interestingly, from \eqref{eqn:UredS} and \eqref{eqn:W_symnormalized}, the sampling-then-reconstruction step is represented as follows:
\begin{equation}
\label{eqn:recon_bpt}
\tilde{\boldx} = \bW' \bI_{\mathcal{V}_1}^\top \bI_{\mathcal{V}_1}\mathbf{G} \boldx.
\end{equation}
This is the same as vertex domain sampling and reconstruction because $\bW'$ and $\bG$ are both graph filters with diagonal graph frequency responses, as in \eqref{eqn:G_USU} and \eqref{eqn:W_symnormalized}, which is illustrated in Fig. \ref{fig:bpt_example}. The sampling-then-reconstruction \eqref{eqn:recon_bpt} is regarded as one branch in a two-channel graph filter bank for a bipartite graph \cite{Narang2012, Narang2013, Sakiya2016a}. The filter bank approach requires at least two such branches to guarantee perfect recovery. In contrast, our generalized graph sampling realizes perfect reconstruction with one branch as long as the DS condition holds (presented in the next section). This may lead to an implementation advantage in practical applications.

\section{Graph Signal Recovery with Subspace Prior}\label{sec:subspace}
In this section, we assume that the signal lies in a known PGS subspace that depends on the given graph. Subsequently, we present two possible correction filters. One is an unconstrained solution that guarantees perfect recovery of the graph signal with an arbitrary choice of the sampling filter. The other is a predefined solution in which a given filter must be used for reconstruction.

\subsection{PGS Subspace}\label{subsec:pgs_space}
We first consider a graph signal subspace that parallels the generation process in SI subspaces shown in \eqref{eqn:xt_si} and \eqref{eq:SI}. As discussed in the previous section, vertex domain sampling is in general a nonuniform sampling operator in contrast to the uniform SI sampling of \eqref{eqn:xt_si}. Hence, we utilize graph frequency domain sampling that mimics ``sampling by modulation'' in \eqref{eqn:Comega}.

In \eqref{eq:SI}, the $2\pi/T$-periodic spectrum of the expansion coefficients $D(e^{j\omega T})$ is multiplied by the (non-periodic) generator $A(\omega)$ to obtain the signal spectrum $X(\omega)$.
We reflect this characteristic in the signal subspace for the graph setting.

The spectrum of the graph considered herein is finite and discrete. Suppose that we have a length $K$ spectrum $\hat{d}(\li)$ ($i = 0, \dots, K-1$, $K\leq N$) as the expansion coefficients. Although its original length is finite, we assume that $\hat{\bmd}$ is periodic beyond $i \ge K$, as in \eqref{eq:SI}, i.e.,
\begin{equation}
\label{eqn:dmodK}
\hat{d}(\li) = \hat{d}(\la_{i\text{ mod }K}).
\end{equation}
Under this assumption, we can naturally define the signal subspace for graph signals as a counterpart of the SI subspace, as follows:
\begin{definition}\label{def:pgs_space}
A PGS subspace of a given graph $\mathcal{G}$ is a space of graph signals that can be expressed as a GFT spectrum filtered by a given generator:
\begin{equation}
\mathcal{X}_{\emph{PGS}} = \left\{x[n] \left| x[n] = \sum_{i=0}^{N-1} \hat{d}(\la_{i \text{\emph{ mod }} K})
A(\li) u_{i}[n]\right.\right\},\label{eqn:pgs_subspace}
\end{equation}
where $A(\li)$ is the graph frequency domain response of the generator and $\hat{d}(\la_i)$ is an expansion coefficient.
\end{definition}
\noindent
This signal can be represented in the following matrix form:
\begin{equation}
\label{eqn:f_subspace}
\boldx :=  \bA \hat{\bmd} = \bU A(\bm{\Lambda})\bD_{\text{samp}}^\top\hat{\bmd}
\end{equation}
where $\hat{\bmd} := [\hat{d}(\la_0), \dots, \hat{d}(\la_{K-1})]^\top$.

Bandlimited graph signals are a special case of signals in a PGS subspace. Suppose that $A(\li)$ is a bandlimiting low-pass filter $G_{\text{BL}, K}(\li)$, i.e.,
\begin{equation}
G_{\text{BL}, K}(\li) = \begin{cases}
1      & i \in [0, K-1], \\
0      & \text{otherwise}.
\end{cases}\label{eqn:GK}
\end{equation}
The graph signal $\boldx$ generated by \eqref{eqn:f_subspace} completely maintains $\hat{\bmd}$, i.e., is $K$-bandlimited under the GFT basis $\bU$. However, the graph signal generated by \eqref{eqn:pgs_subspace} with an arbitrary $A(\li)$ is in general not necessarily bandlimited; thus, our generalized sampling introduced in the following sections allows for reconstruction of non-bandlimited graph signals.

A similar assumption as \eqref{eqn:f_subspace} without periodicity is given in \cite{Dong2016}, where $\hat{\bmd}$ is assumed to be a random variable. Our PGS assumption includes this definition: In particular, if $K=N$ and $\hat{\bmd}$ is a random variable following a zero-mean multivariate Gaussian distribution, then the subspace defined by Definition \ref{def:pgs_space} coincides with that used in \cite{Dong2016}. However, note that we impose no constraint on the generator function or the expansion coefficients. Thus, Definition \ref{def:pgs_space} treats a broader class of graph signals than the above.

Suppose that $T$ in \eqref{eq:SI} is a positive integer, i.e., the spectra $D(e^{j\omega T})$ are repeated $T$ times within $\omega \in [0, 2\pi]$, and $A(\omega)$ in \eqref{eq:SI} has support $\omega \in [0, 2\pi]$. In this case, a sequence $X[i] = \left.D(e^{j\omega T}) A(\omega)\right|_{\omega = 2\pi i/N}$ ($i = 0, \dots, N-1$) corresponds to the DFT spectrum of length $N$. Therefore, this $X[i]$ can be regarded as a graph signal spectrum in a PGS subspace when $\bU^*$ is the DFT matrix, e.g., the graph $\mathcal{G}$ is a circular graph.

\subsection{Unconstrained Case}\label{subsec:subspace_unc}
Our solutions for generalized graph sampling can be defined following the general Hilbert space results of Section \ref{sec:gensamp}. Owing to the definition of the PGS subspace and sampling in the graph frequency domain, the sampling, correction, and reconstruction filters can all be implemented in the graph frequency domain.

\subsubsection{Recovery Filters}
For the unconstrained solution, we use a reconstruction filter $W(\li) = A(\li)$ in \eqref{eqn:pgs_subspace}. Suppose that the DS condition \eqref{eqn:DS_condition} is satisfied for the signal and sampling subspaces. Following the expression in \eqref{eqn:x_ss_ds_h}, the signal recovery is given as follows:
\begin{equation}
\label{ }
\begin{split}
\tilde{\boldx} &= \bA (\bS^* \bA)^{-1} \bS^* \boldx\\
& = \bA (\bS^* \bA)^{-1} \bS^* \bA \hat{\bmd}\\
&=  \bA \hat{\bmd} = \boldx,
\end{split}
\end{equation}
where the correction filter is
\begin{equation}
\label{eqn:H_subspace_unc}
\bH = (\bS^* \bA)^{-1}.
\end{equation}
Its graph frequency response is
\begin{equation}
\label{eqn:bH}
H(\lambda_i) = \frac{1}{\tilde{R}_{SA}(\la_i)}
\end{equation}
where
\begin{equation}
\label{eqn:graph_cc}
\tilde{R}_{SA}(\la_i) := \sum_{\ell} S(\la_{i + K\ell})A(\la_{i + K\ell}).
\end{equation}
The inverse of $\tilde{R}_{SA}(\la_i)$ is well defined under the DS condition. Note the similarity with \eqref{eqn:Homega}.

The solution for the LS and MX strategies when $\mA$ and $\mS$ intersect can be derived from \eqref{eqn:x_ss_ls_h} as follows:
\begin{equation}
\label{eqn:fhat_subspace_unc}
\boldxhat = \bA (\bS^* \bA)^{\dagger} \bS^* \boldx.
\end{equation}
The correction filter in this case is $\bH = (\bS^* \bA)^{\dagger}$, and it has the same graph frequency response as \eqref{eqn:bH} but with $H(\lambda_i) = 0$ for $\li$ with $\tilde{R}_{SA}(\la_i) = 0$.

\subsubsection{Special Cases}\label{subsubsec:ss_unc_bpt}
Suppose that both the generator and sampling filters are $A(\bLam) = S(\bLam) = G_{\text{BL}, K}(\bLam)$ in \eqref{eqn:GK}.
Subsequently, $H(\li) = 1$ and no correction filter is required. This is equivalent to the perfect recovery condition for bandlimited graph signals using graph frequency domain sampling \cite{Tanaka2018}.

Another interesting case is the bipartite graph introduced in Section \ref{subsec:bpt}. For example, suppose that $S(\li)$ in \eqref{eqn:G_USU} is $G_{\text{BL}, N/2}(\li)$ and the generator is $A(\li) = G_{\text{IR}}(\li)$ with
\begin{equation}
\label{eqn:G_IR}
G_{\text{IR}}(\li) = \begin{cases}
1      & \la_0 \leq \li \leq 2/\la_{\max}, \\
-\frac{2\li}{\la_{\max}}    & \li > 2/\la_{\max},
\end{cases}
\end{equation}
where the correction filter again becomes $H(\li) = 1$; therefore, $W(\li) = A(\li) =G_{\text{IR}}(\li)$. This implies that a non-bandlimited graph signal can be perfectly reconstructed from bandlimited measurements by applying the same filtering as in the generation process without an explicit correction filter.

In addition, as mentioned in \eqref{eqn:recon_bpt}, our sampling and recovery can be represented by using vertex domain sampling for bipartite graphs. Existing graph filter banks for bipartite graphs, e.g., \cite{Narang2012, Narang2013, Sakiya2016a}, may require length-$N$ coefficients (for maximally decimated transforms) for reconstruction on the synthesis side. Instead, our framework only needs one channel which requires a length-$K$ spectrum for recovery, as demonstrated in Section \ref{subsec:exp_bpt}. If the generator function $A(\li)$ of a given signal is losslessly encoded and sent along with the spectrum, we can reconstruct the original signal. This may be regarded as a one-branch compression of a graph signal.

\subsection{Predefined Case}
Suppose that the reconstruction filter $W(\li)$ is predefined. The reconstructed signal $\tilde{\boldx}$ will in general be different from $\boldx$ in this case. As in the unconstrained setting introduced in the previous subsection, the correction transforms in our framework are given through graph spectral filters.

If $\mA$ and $\mS$ satisfy the DS condition in \eqref{eqn:DS_condition}, the solution in \eqref{eqn:H_ss_ds_h} reduces to the following:
\begin{equation}
\label{ }
\bH = (\bW^* \bW)^{-1} \bW^* \bA (\bS^* \bA)^{-1}.
\end{equation}
The corresponding graph filter is
\begin{equation}
\label{eqn:H_ss_const}
H(\la_i) = \frac{\tilde{R}_{WA}(\la_i)}{\tilde{R}_{SA}(\la_i)\tilde{R}_{WW}(\la_i)}.
\end{equation}
If $W(\li) = A(\li)$, the response above is identical to that of the unconstrained case shown in \eqref{eqn:bH}.

Without the DS condition, we can apply the LS and MX strategies. The LS solution is
\begin{equation}
\label{eqn:ss_predef_cls}
\tilde{\boldx} = \bW (\bS^* \bW)^{\dagger} \bS^* \boldx,
\end{equation}
where the correction filter $\bH = (\bS^* \bW)^{\dagger}$
has spectral response
\begin{equation}
\label{eqn:bH_subspace_cls}
H(\lambda_i) = \begin{cases}
\frac{1}{\tilde{R}_{SW}(\la_i)}  & \tilde{R}_{SW}(\la_i) \neq 0,\\
0      & \text{otherwise}.
\end{cases}
\end{equation}

The MX solution becomes
\begin{equation}
\label{ }
\tilde{\boldx} = \bW (\bW^* \bW)^{-1} \bW^* \bA (\bS^* \bA)^{\dagger} \bS^* \boldx,
\end{equation}
with
\begin{equation}
\label{ }
\bH = (\bW^* \bW)^{-1} \bW^* \bA (\bS^* \bA)^{\dagger}.
\end{equation}
The spectral response of the filter now is the same as that in \eqref{eqn:H_ss_const} but $H(\li) = 0$ if the denominator is zero.

The graph correction filters are summarized in Table \ref{tb:SI_filters}. The table demonstrates nicely the similarities with SI sampling.

\section{Graph Signal Recovery with Smoothness Prior}\label{sec:smoothness}
The subspace prior introduced in the previous section enables the input graph signal to be recovered perfectly; however, it requires full knowledge of the given graph and generator. In this section, we consider a less restrictive assumption. We still assume that the GFT basis $\bU$ is given; however, the generator function $A(\li)$ is unknown.

We assume that the graph signal is smooth on the given graph where smoothness is measured by the signal energy in the high graph-frequency components as in the SI setting \eqref{eqn:high_freq_energy_SI}. Although several possible operators exist for measuring signal smoothness on a graph \cite{Shuman2013}, we consider a simple quadratic form of $\boldx$:
\begin{equation}
\|\bV \boldx\|^2_2 = \boldx^* \bV^{2} \boldx = \sum_{i = 0}^{N-1} V^2(\li)|\hat{x}(\li)|^2
\label{eqn:smoothness_quadratic}
\end{equation}
where $\bV := \bU V(\bLam) \bU^*$ is an arbitrary graph filter with spectral response $V(\li)$.
The smoothness condition is given by $\|\bV \boldx\|^2_2 \leq \rho^2$ for some constant $\rho$.
This can be seen as a generalization of the bandlimitedness of graph signals, which has been widely studied \cite{Anis2016, Chen2015, Puy2018, Tsitsv2016, Pesens2008, Pesens2010}, because a bandlimited graph signal corresponds to $\rho=0$ for a high-pass filter $\bV = \bU (\bI - G_{K, \text{BL}}(\bLam)) \bU^*$. In addition, if we assume $\bV = \bL^{1/2}$, then $\|\bV \boldx\|^2_2 = \boldx^* \bL \boldx$, which is a Laplacian quadratic form also used extensively in the literature. Hereinafter, for simplicity, we assume $V(\li) \neq 0$ for all $i$.

For the unconstrained case, the LS recovery is given from \eqref{eqn:x_sm_unc} as follows:
\begin{equation}
\label{eqn:smooth_unc}
\tilde{\boldx}  = \widetilde{\bW} (\bS^* \widetilde{\bW})^{-1}\bS^* \boldx,
\end{equation}
where $\widetilde{\bW} = (\bV^* \bV)^{-1} \bS = \bU V^2(\bLam) \bU^*$ and $\bS^* \bV^{-2} \bS = \bD_{\text{samp}} S^2(\bLam) V^{-2}(\bLam) \bD_{\text{samp}}^\top$ is invertible if $\tilde{R}_{SS}(\li) \neq 0$ for all $i$. This results in $\bH = (\bS^* \widetilde{\bW})^{-1}$, where the spectral response is
\begin{equation}
\label{eqn:smooth_unc_ls}
H(\lambda_i) = \frac{1}{\tilde{R}_{S\widetilde{W}}(\la_i)}.
\end{equation}
The MX solution coincides with \eqref{eqn:smooth_unc_ls} as in the SI solution.

We next consider the predefined case. For the LS criterion, the solution in Hilbert space \eqref{eqn:sm_costfunc_pd} reduces to the constrained LS solution with a subspace prior \eqref{eqn:bH_subspace_cls}: This does not depend on $V(\li)$, i.e., the smoothness prior does not affect the solution.

The MX solution can be obtained from \eqref{eqn:x_sm_pd_mx}:
\begin{equation}
\label{eqn:sm_pd_recon}
\tilde{\boldx} = \bW (\bW^* \bW)^{-1} \bW^* \widetilde{\bW} (\bS^* \widetilde{\bW})^{-1} \bS^* \boldx.
\end{equation}
This leads to
\begin{equation}
\label{ }
\bH = (\bW^* \bW)^{-1} \bW^* \widetilde{\bW} (\bS^* \widetilde{\bW})^{-1},
\end{equation}
wtih spectral response
\begin{equation}
\label{eqn:H_smooth_mx}
H(\la_i) = \frac{\tilde{R}_{W\widetilde{W}}(\la_i)}{\tilde{R}_{S\widetilde{W}}(\la_i)\tilde{R}_{WW}(\la_i)}.
\end{equation}
The smoothness prior $V(\li)$ is incorporated appropriately in the correction filter, in contrast to the LS solution.

These correction filters are summarized in Table \ref{tb:SI_filters}.

\section{Comparison with Existing Studies}\label{sec:relationship}
\subsection{Computational Complexity}\label{subsec:complexity}
Here, we compare the computational complexities of vertex and graph frequency domain sampling required for sampling, recovery, and preprocessing.

\subsubsection{Sampling}
For spectral domain sampling, its complexity is $\mathcal{O}(NK)$ because the sampling matrix $\bS^*$ has a size of $K \times N$. However, its complexity is reduced to $\mathcal{O}(N)$ if we already have GFT coefficients $\hat{\boldx}$ because the sampling filter response $S(\bLam)$ is diagonal. Typically, the GFT is necessary to perform only once even when the sampling ratio or filter is changed.

For vertex domain sampling, the subsampling itself in \eqref{eqn:vertex_samp} only picks up elements specified by the sampling set $\mT$: Its computation cost is negligible. In contrast, a sampling filter in the vertex domain requires $\mathcal{O}(N^2)$ complexity in general. This can be reduced using a localized filter: $P$-hop filtering requires $\mathcal{O}(P|\mathcal{E}|)$ complexity \cite{Hammon2011}.

\subsubsection{Recovery}
As mentioned in the above two sections, all correction and reconstruction filters in our generalized sampling framework have diagonal graph frequency responses, the computational complexity of which is $\mathcal{O}(K)$. When we change the sampling filter, the response of the correction filter can be immediately calculated because it is diagonal. An additional complexity is required for the inverse GFT (where its complexity depends on the GFT used) if reconstructed vertex domain signals are required.

For signal recovery using vertex domain sampling, a matrix of size $N\times K$ is multiplied by the sampled coefficients, the complexity of which is $\mathcal{O}(NK)$. Note that the reconstruction matrix depends on $\mT$: If the sampling set $\mT$ or sampling rate $|\mT|$ is changed, we have to re-calculate the entire reconstruction matrix even if the graph is the same. This calculation typically requires the inversion of a $K \times K$ matrix, e.g., \cite{Chen2015, Anis2016, Sakiya2019a}, and may lead to $\mathcal{O}(K^3)$ complexity.

\subsubsection{Preprocessing}
Graph frequency domain sampling considered in this paper requires the GFT matrix, and hence we need to compute an eigendecomposition of the variation operator. This typically requires $\mathcal{O}(N^3)$ complexity (whereas we can use several fast computation methods of GFT or spectral decomposition of the graph variation operator such as in \cite{LeMag2018, Lu2019, Heimow2017, Heimow2018, Heimow2018a}). It is important to note that the same graph is often used numerous times. In such a case, we only need to calculate the GFT basis once, and can reuse it even when we change the sampling rate $M = N/K$ or the sampling filter.

Vertex domain sampling always requires computing the best $\mT$ from a given graph. The computation complexity highly depends on the sampling set selection methods. Major methods are compared in \cite{Sakiya2019a}. Typically, the complexity depends on $N$, the (assumed) cutoff graph frequency, and the edge density.
\subsection{Literature Review}
In \cite{Chepur2018}, a generalized sampling method for graph signal processing was studied. 
As the results did not assume any particular structure on the input signals and sampling filters, the recovery procedures were in general given by matrix inversions. Here, we focus on a special case of \cite{Chepur2018} that extends SI sampling to the graph setting and enables explicit expressions for the recovery filter in the graph Fourier domain.

Our solution represented in \eqref{eqn:fhat_subspace_unc} allows for a broad choice of $S(\li)$ and $A(\li)$. In particular, $A(\li)$ is not restricted to a bandlimiting operator. If we have $S(\bLam) = A(\bLam) = G_{\text{BL}, K}(\bLam)$, our solution reduces to that of \cite{Chepur2018}, which is equivalent to the sampling theory with graph frequency domain sampling studied in \cite{Tanaka2018}.

For the smoothness prior, if the smoothness is measured by $\bV^2 = \bL + G_{\text{BL}, K}(\bLam)$, i.e.,
\begin{equation}
\label{ }
V(\li) = \begin{cases}
\sqrt{\li + 1} & i \leq K -1\\
\sqrt{\li} & i\geq K
\end{cases}
\end{equation}
and $S(\bLam)= G_{\text{BL}, K}(\bLam)$, our solution also reduces to that introduced in \cite{Chepur2018}: The reconstruction and correction filters as shown in \eqref{eqn:smooth_unc} exhibit the following form:
\begin{equation}
\label{ }
W(\li) = \frac{1}{\li + 1},\quad H(\li) = \li + 1
\end{equation}
for $\li \in [\la_0, \la_{K-1}]$, and $W(\li) = H(\li) = 0$ otherwise. This is a special case of \cite{Chepur2018} where the correction and reconstruction operators can be represented as spectral filters.

As mentioned in Section \ref{subsec:pgs_space}, many studies on graph sampling theory such as \cite{Anis2016, Chen2015, Marque2016} implicitly assume that the graph signal lies in the PGS subspace with a typical generator function $A(\li) = G_{\text{BL}, K}(\li)$. While their subspace is a special case of the PGS assumption, the sampling matrices are different from that in \eqref{eqn:S}. As described in Definition \ref{def:GD_vertex}, the simple subsampling $\bS^* = \bI_\mT$ has been used in many studies on graph sampling \cite{Anis2016, Chen2015}. In \cite{Marque2016}, aggregation sampling was used.
However, its sampling matrix in \eqref{eqn:G_agg} does not in general have a corresponding sampling expression in the graph frequency domain as that in \eqref{eqn:gft_sampling}.
This results in the requirement of matrix inversion even for recovering the bandlimited graph signal although the signal lies in a PGS subspace.

In summary, most studies on graph sampling theory require inversion of the sampling operator for their reconstruction framework. Moreover, they focused on the design problem for the nonuniform sampling matrix $\bI_\mT$ that, for example, maximizes the bandwidth with perfect recovery. In contrast, frequency domain sampling is utilized in this study as a counterpart of ``sampling by modulation'' in SI spaces, thus resulting in a symmetric structure, i.e., both the sampling and reconstruction steps can be represented as similar sampling and filtering operations. We also allow for a broader set of input signals and design criteria.

\section{Graph Signal Recovery Experiments}\label{sec:exp}
In this section, we validate the proposed generalized sampling through signal recovery experiments. First, we demonstrate that the correction and recovery filters described in Sections \ref{sec:subspace} and \ref{sec:smoothness} reduce the reconstruction error of non-bandlimited graph signals compared to the bandlimited sampling in the graph frequency domain \cite{Tanaka2018} and FastGSSS \cite{Sakiya2019a}, which is a state-of-the-art vertex domain sampling method. This reveals the MSE improvements of our generalized sampling for non-bandlimited graph signals over bandlimited or smoothness-based reconstruction. Sampling under the bipartite case presented in Section \ref{subsec:bpt} is then conducted in which full-band graph signals are almost perfectly recovered with one branch of sampling and reconstruction, even without calculating the GFT matrix.

\begin{table*}[t]
\caption{Average MSEs of Reconstructed Signals after 1000 Independent Runs (in Decibels). Columns with BL Refer to Bandlimited Sampling, and Those with Non-BL Refer to Non-Bandlimited Sampling.}\label{tab:experiments}
\centering
\begin{tabular}{c|c||r|r|r|r||r|r|r|r}
\hline
& & \multicolumn{4}{c||}{Generator function \#1 \eqref{eqn:A_exp}} & \multicolumn{4}{c}{Generator function \#2 \eqref{eqn:A_exp2}}\\\cline{3-10}
&Solution/&\multicolumn{2}{c|}{Noiseless signals}&\multicolumn{2}{c||}{Noisy signals}&\multicolumn{2}{c|}{Noiseless signals}&\multicolumn{2}{c}{Noisy signals}\\\cline{3-10}
Prior&Strategy&\multicolumn{1}{c|}{BL}&\multicolumn{1}{c|}{Non-BL}&\multicolumn{1}{c|}{BL}&\multicolumn{1}{c||}{Non-BL} & \multicolumn{1}{c|}{BL}&\multicolumn{1}{c|}{Non-BL}&\multicolumn{1}{c|}{BL}&\multicolumn{1}{c}{Non-BL}\\\hline
Subspace&Unconstrained&-296.65&-302.56&-9.66&-10.54 & -299.84&-303.80&-10.22&-10.35\\
&Predefined: DS and MX&-18.52&-18.52&-9.00&-9.92&-12.16&-12.16&-7.91&-8.29\\
&Predefined: LS*&-13.82&-18.45&-7.84&-9.92&-7.25&-12.10&-5.12&-8.26\\\hline
Smoothness&Unconstrained&-3.99&-20.30&-3.14&-10.06&-8.80&-18.67&-6.58&-9.87\\
&Predefined: MX&-5.14&-18.49&-3.95&-9.92&-8.33&-12.15&-6.14&-8.28\\\hline\hline
\multicolumn{2}{c||}{BL sampling and reconstruction \cite{Tanaka2018}}&\multicolumn{2}{c|}{-3.99}&\multicolumn{2}{c||}{-3.86}&\multicolumn{2}{c|}{-8.80}&\multicolumn{2}{c}{-8.41}\\
\multicolumn{2}{c||}{FastGSSS \cite{Sakiya2019a}}&\multicolumn{2}{c|}{-2.31}&\multicolumn{2}{c||}{-1.85}&\multicolumn{2}{c|}{-7.18}&\multicolumn{2}{c}{-5.88}\\\hline
\multicolumn{10}{l}{* Same as the predefined solution for smoothness prior with LS strategy}\\
\end{tabular}
\end{table*}%

\subsection{Recovery Experiments for Bandlimited and Non-bandlimited Sampling}
We first conduct signal recovery experiments following our generalized sampling framework shown in Fig. \ref{fig:generalized_sampling}. Because we have two choices of priors, i.e., subspace and smoothness priors; three strategies, i.e., DS, LS, and MX; and two possible reconstruction filters, i.e., unconstrained or predefined, we compare these settings throughly in this simulation. In addition, two sampling filters are considered: bandlimited and non-bandlimited. The bandlimited sampling corresponds to a graph sampling theory described in \cite{Tanaka2018}, whereas our generalized sampling recovers the original full-band graph signals after bandlimited sampling. As mentioned previously, we allow for non-bandlimited sampling filters to obtain the sampled signal $\hat{\bmc}$ while still guaranteeing perfect recovery under certain conditions. Here, we also demonstrate this property in this subsection.

The graph used is a random sensor graph with $N = 256$. We downsampled the input signal by a factor of two such that $K=32$.
We used the following functions:

\begin{itemize}
  \item Generator function.
\begin{numcases}{A(\li) =}
1-\li/(\lambda_{\max}+\epsilon) & \text{Function \#1} \label{eqn:A_exp}\\
\exp(-1.5\li/\lambda_{\max})& \text{Function \#2} \label{eqn:A_exp2}
\end{numcases}
  \item Sampling functions.
\begin{equation}
S(\li) =\begin{cases}
G_{\text{BL},K}(\li)&\text{for bandlimited sampling}\\
G_{\text{IR}}(\li)&\text{for non-bandlimited sampling.}
\end{cases}\label{eqn:exp_samplingfilters}
\end{equation}
  \item Reconstruction function (used only for the predefined solutions).
\begin{equation}
\label{ }
W(\li) = \cos\left(\frac{\pi}{2} \cdot \frac{\li}{\lambda_{\max}+\epsilon}\right)
\end{equation}
  \item Smoothness function (used only for the smoothness prior).
\begin{equation}
\label{ }
V(\li) = \li/\lambda_{\max} + 1.
\end{equation}
\end{itemize}
We set $\epsilon = 0.1$. All functions are visualized in Fig. \ref{fig:functions_exp}. It is worth noting that both $A(\li)$ are not bandlimited; therefore, the original signal retains its full band. The generator function in \eqref{eqn:A_exp2} is smoother than \eqref{eqn:A_exp}: Here, \eqref{eqn:A_exp2} decays more rapidly than \eqref{eqn:A_exp} when $\la$ increases. Each element in the expansion coefficients $\hat{\bmd}$ is a random variable drawn from $\mathcal{N}(1, 1)$. Examples of $\boldx$ generated by $A(\li)$ in \eqref{eqn:A_exp} are shown in Fig. \ref{fig:experiments_samp}(a). For the proposed sampling, we perform two samplings to highlight the difference between the proposed sampling and the bandlimited sampling, as in \eqref{eqn:exp_samplingfilters}.

For comparison, we applied bandlimited signal recovery with GFT domain sampling \cite{Tanaka2018}: $S(\li) = W(\li) = G_{\text{BL},K}(\li)$ with no correction filter $H(\li) = 1$.
In addition, signal recovery using FastGSSS \cite{Sakiya2019a} is also conducted. FastGSSS assumes that $\boldx$ is smooth on a graph. It determines the sampling set $\mT$ and applies a recovery based on a polynomial graph filter with a given kernel. FastGSSS perfectly recovers the original signal when the polynomial order approaches infinity if the signal is bandlimited.

We applied $1,000$ independent runs and calculated the average MSEs. We then repeated the experiments using zero-mean Gaussian noise with variance $\sigma^2 = 0.1$ added to $\boldx$.

\begin{figure}[tp]
   \centering
   \includegraphics[width = \linewidth]{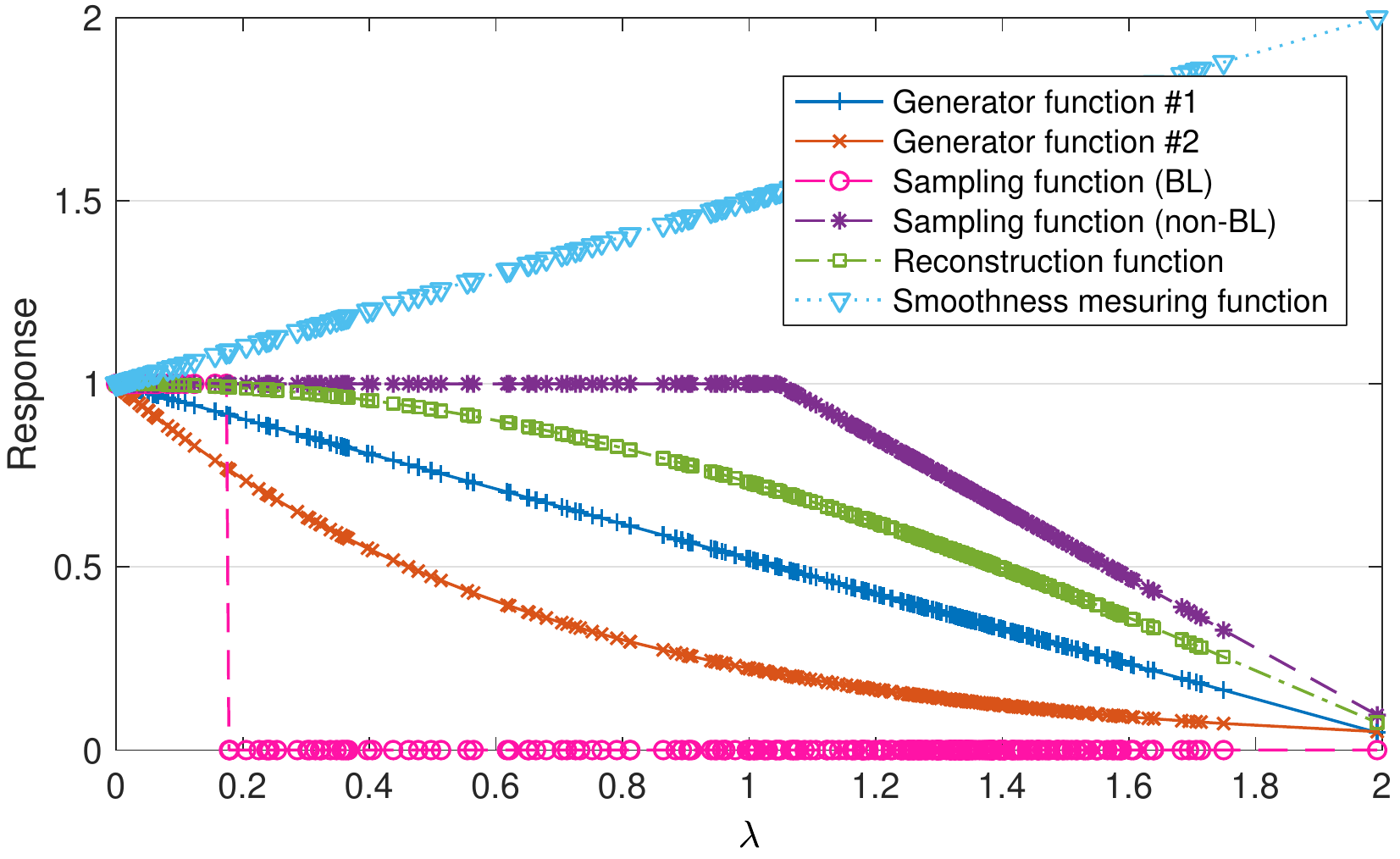} 
   \caption{Spectral responses of several functions used for the experiments.}
   \label{fig:functions_exp}
\end{figure}

\begin{figure*}[tp]
\centering
\subfigure[][Original]{\includegraphics[width=.19\linewidth]{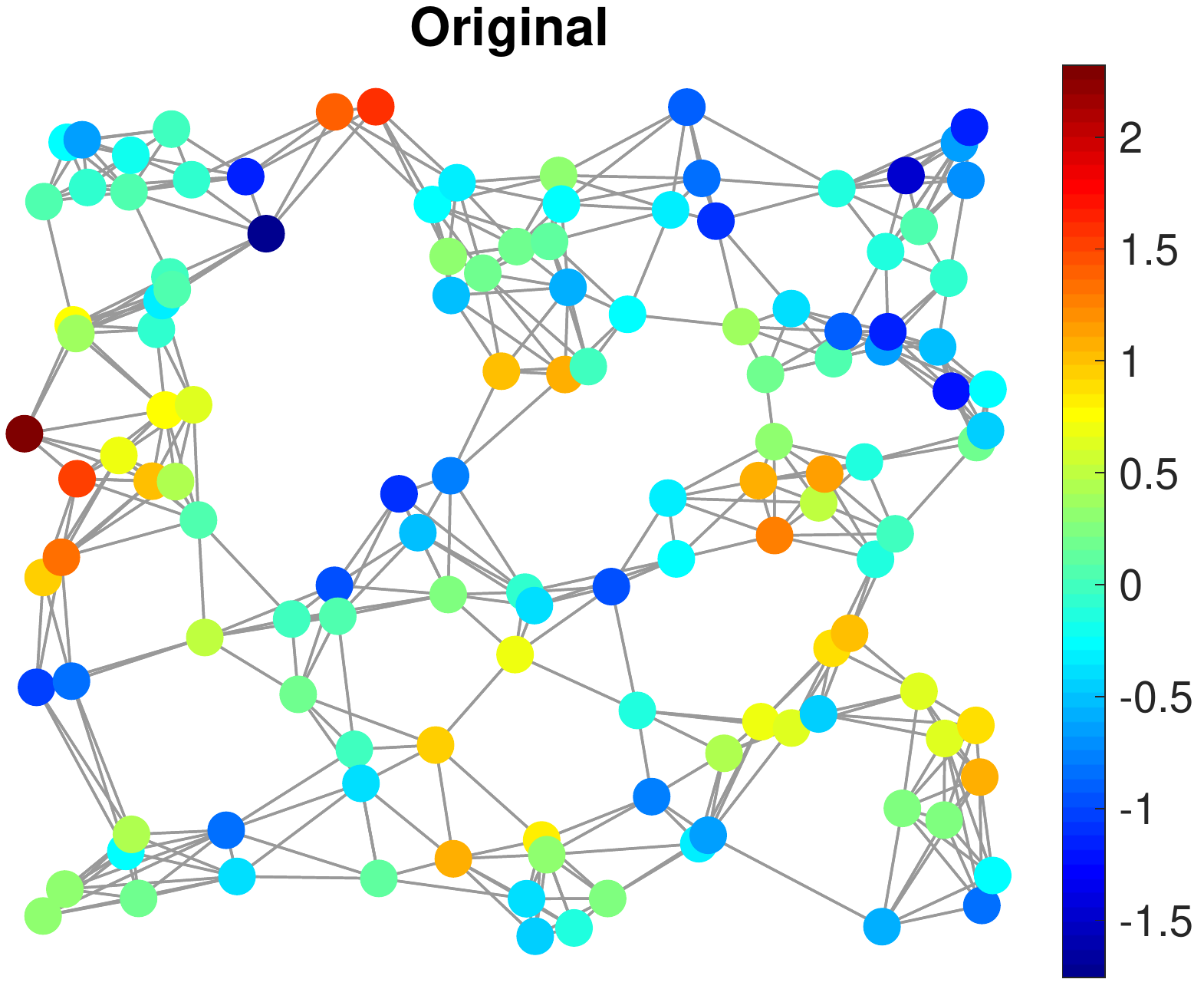}}\ 
\subfigure[][BL samp. + recon. \cite{Tanaka2018}]{\includegraphics[width=.19\linewidth]{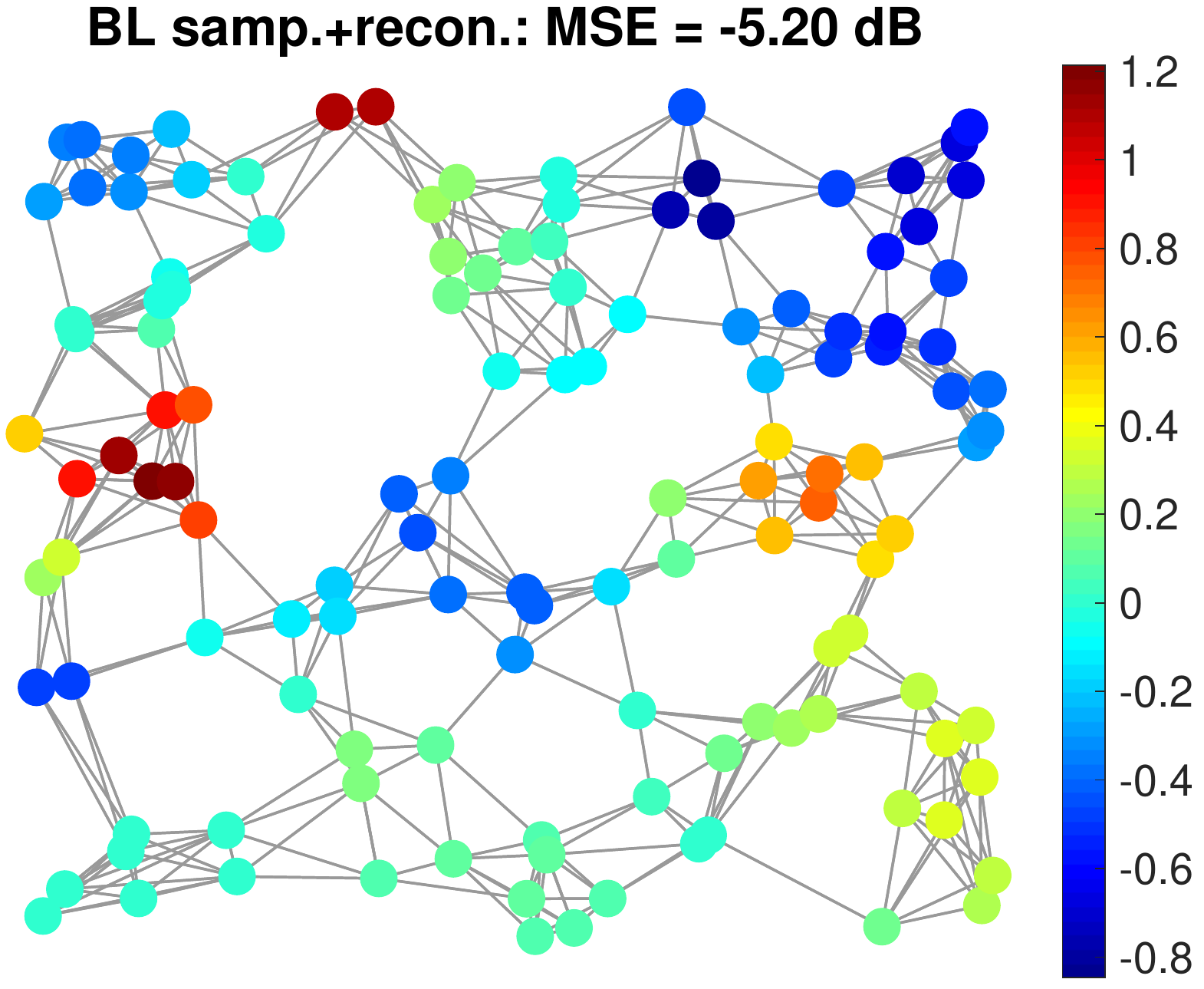}}\ 
\subfigure[][FastGSSS \cite{Sakiya2019a}]{\includegraphics[width=.19\linewidth]{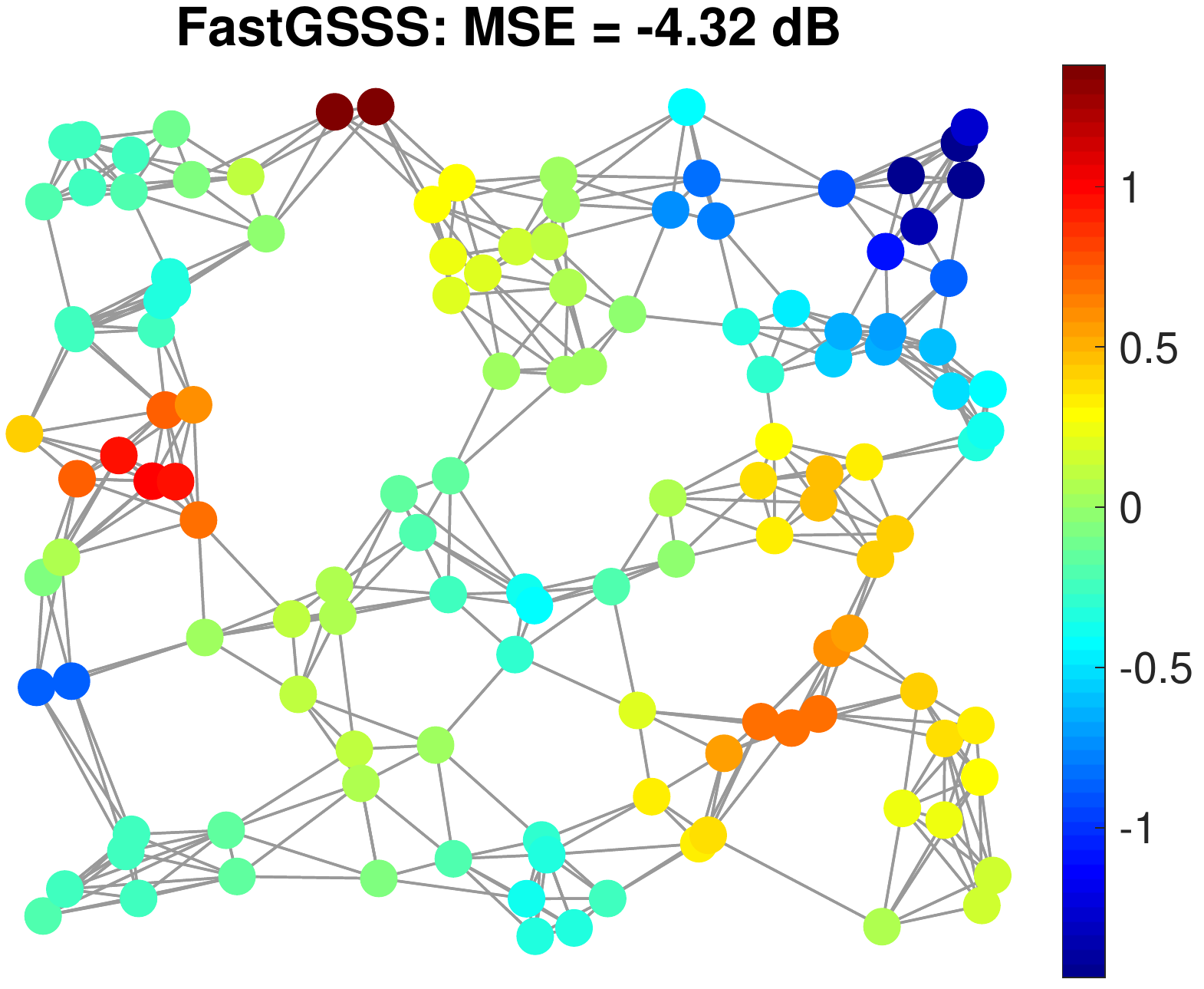}}\\
\subfigure[][Reconstructed: SS UNC (BL sampling)]{\includegraphics[width=.19\linewidth]{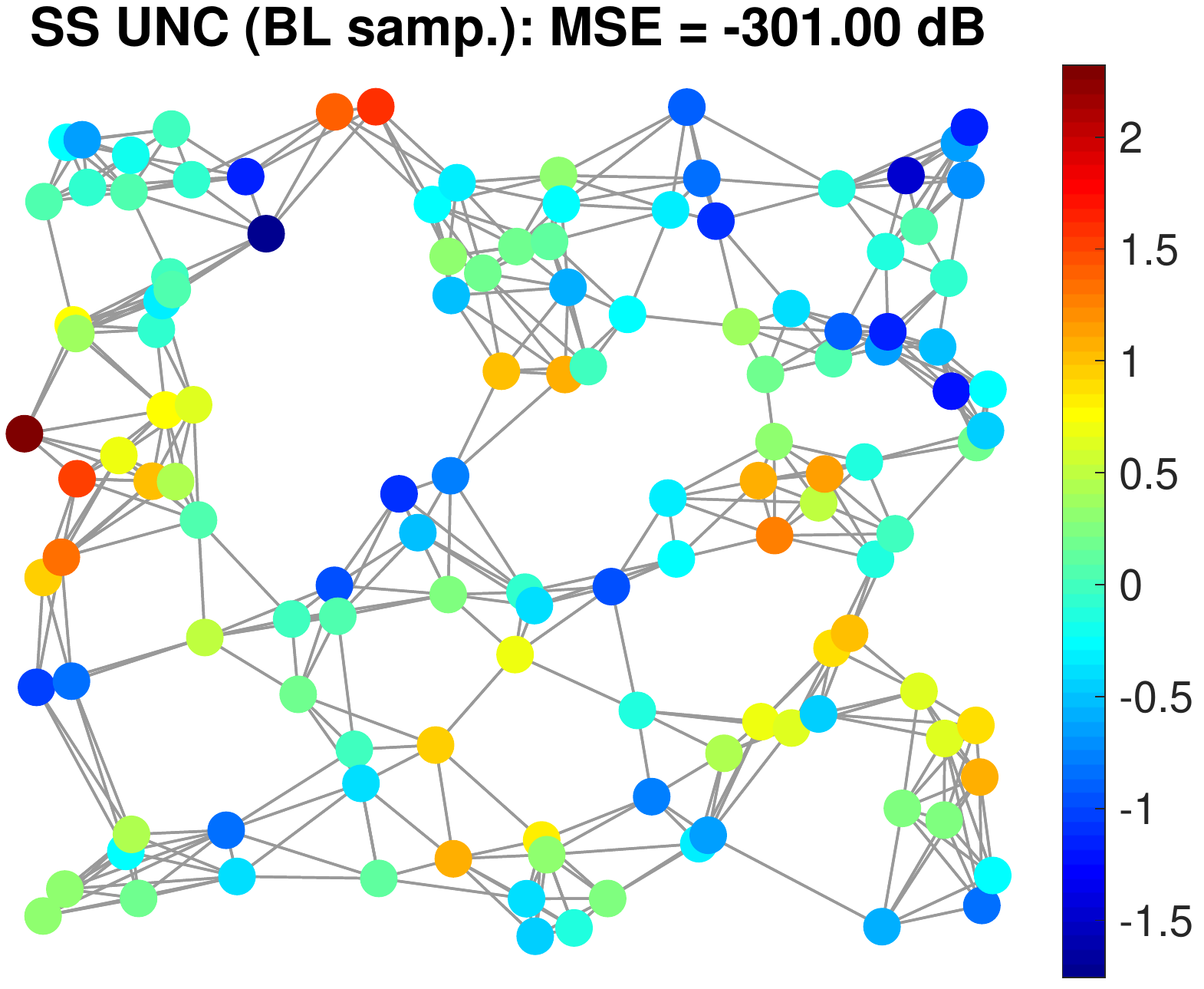}}\ 
\subfigure[][Reconstructed: SS PD DS/MX (BL sampling)]{\includegraphics[width=.19\linewidth]{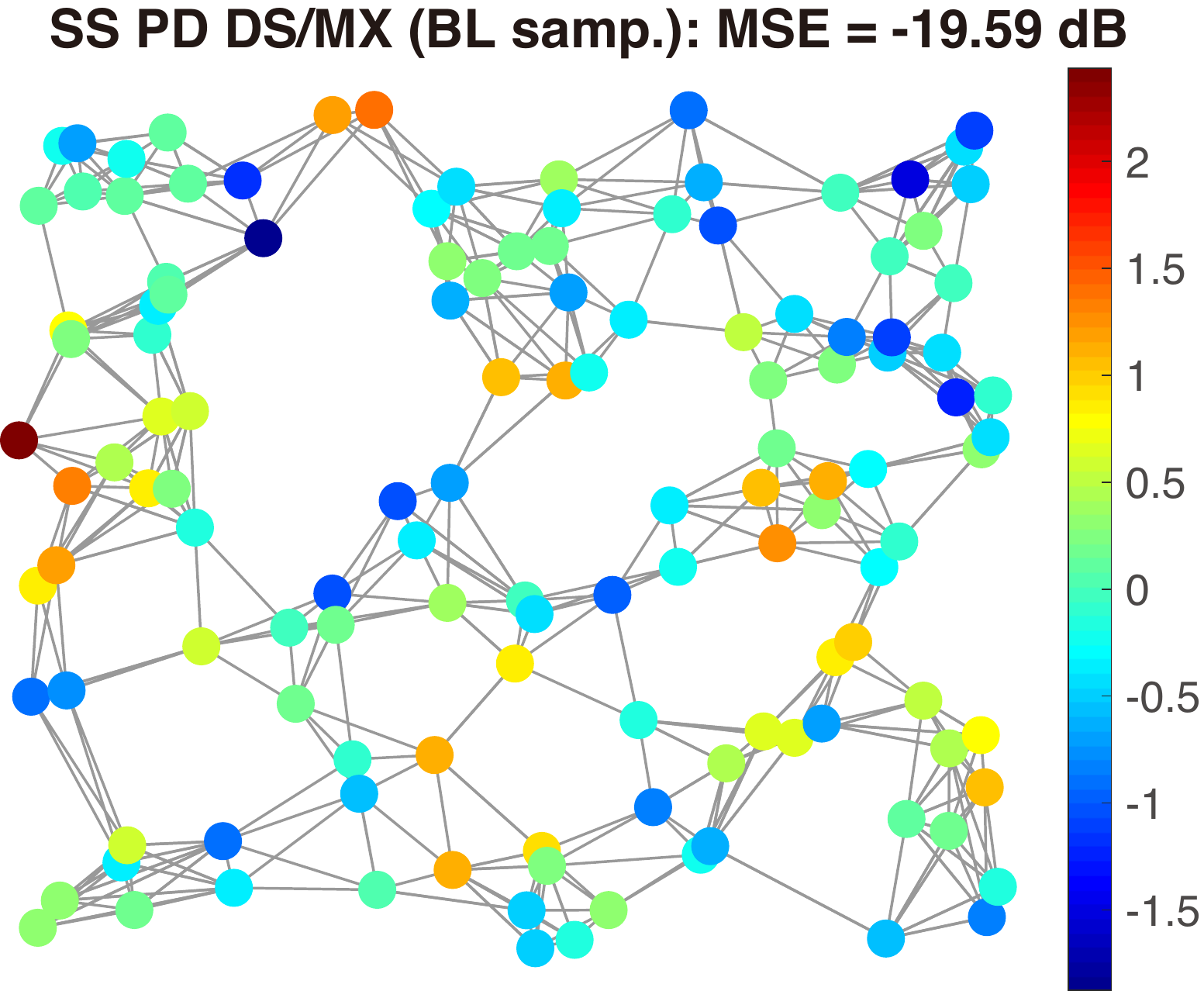}}\ 
\subfigure[][Reconstructed: SS PD LS (BL sampling)]{\includegraphics[width=.19\linewidth]{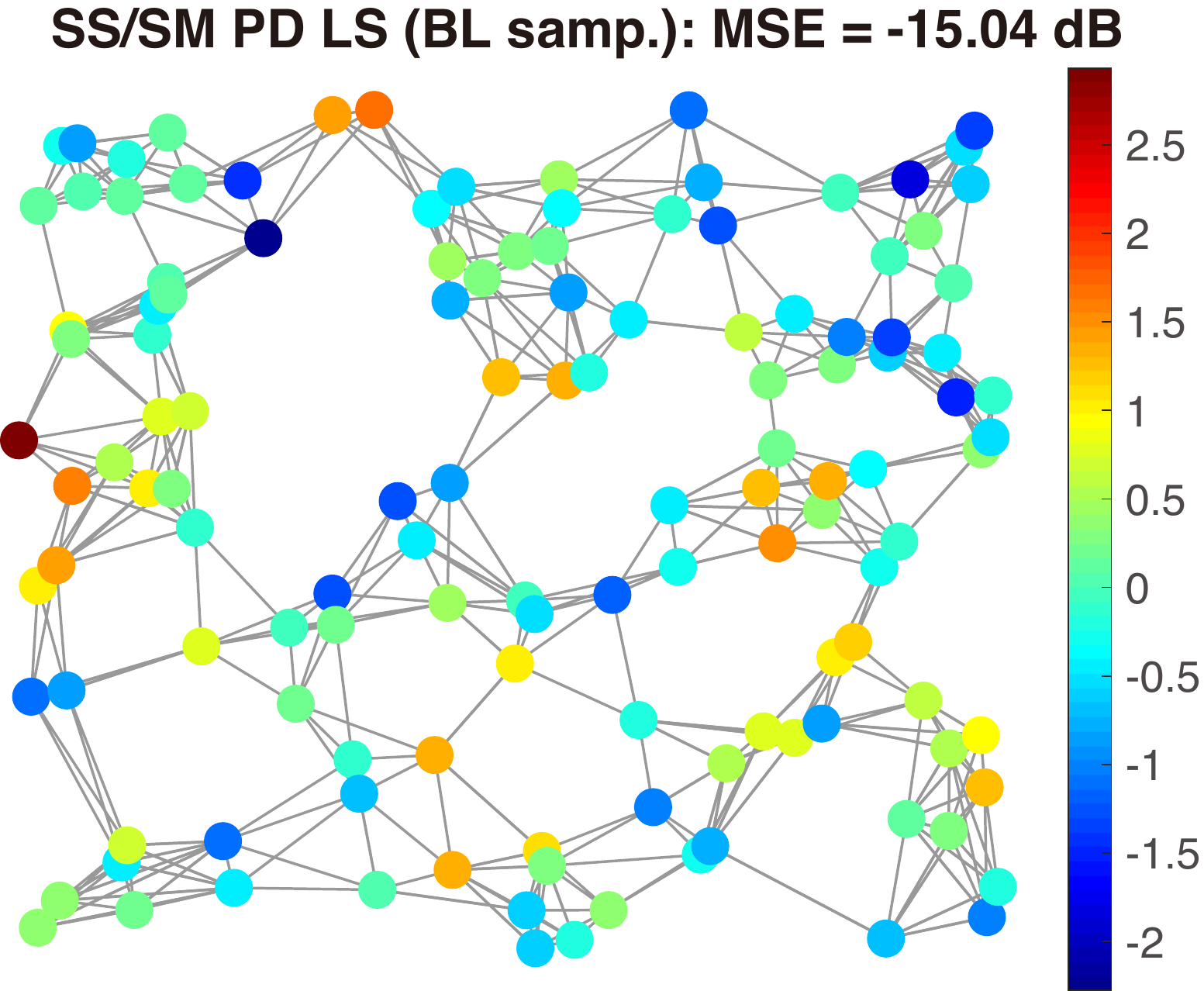}}\ 
\subfigure[][Reconstructed: SM UNC (BL sampling)]{\includegraphics[width=.19\linewidth]{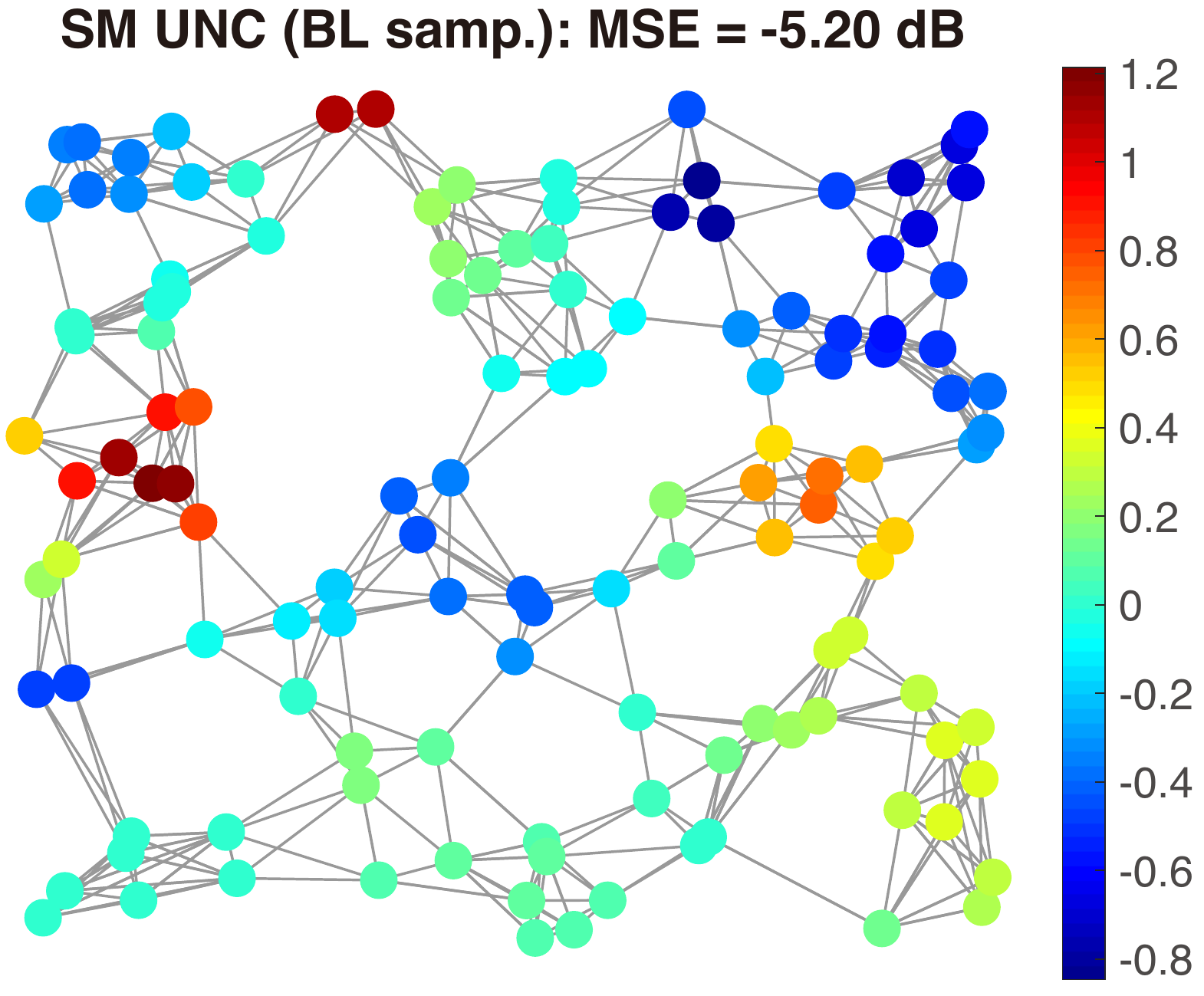}}\ 
\subfigure[][Reconstructed: SM PD MX (BL sampling)]{\includegraphics[width=.19\linewidth]{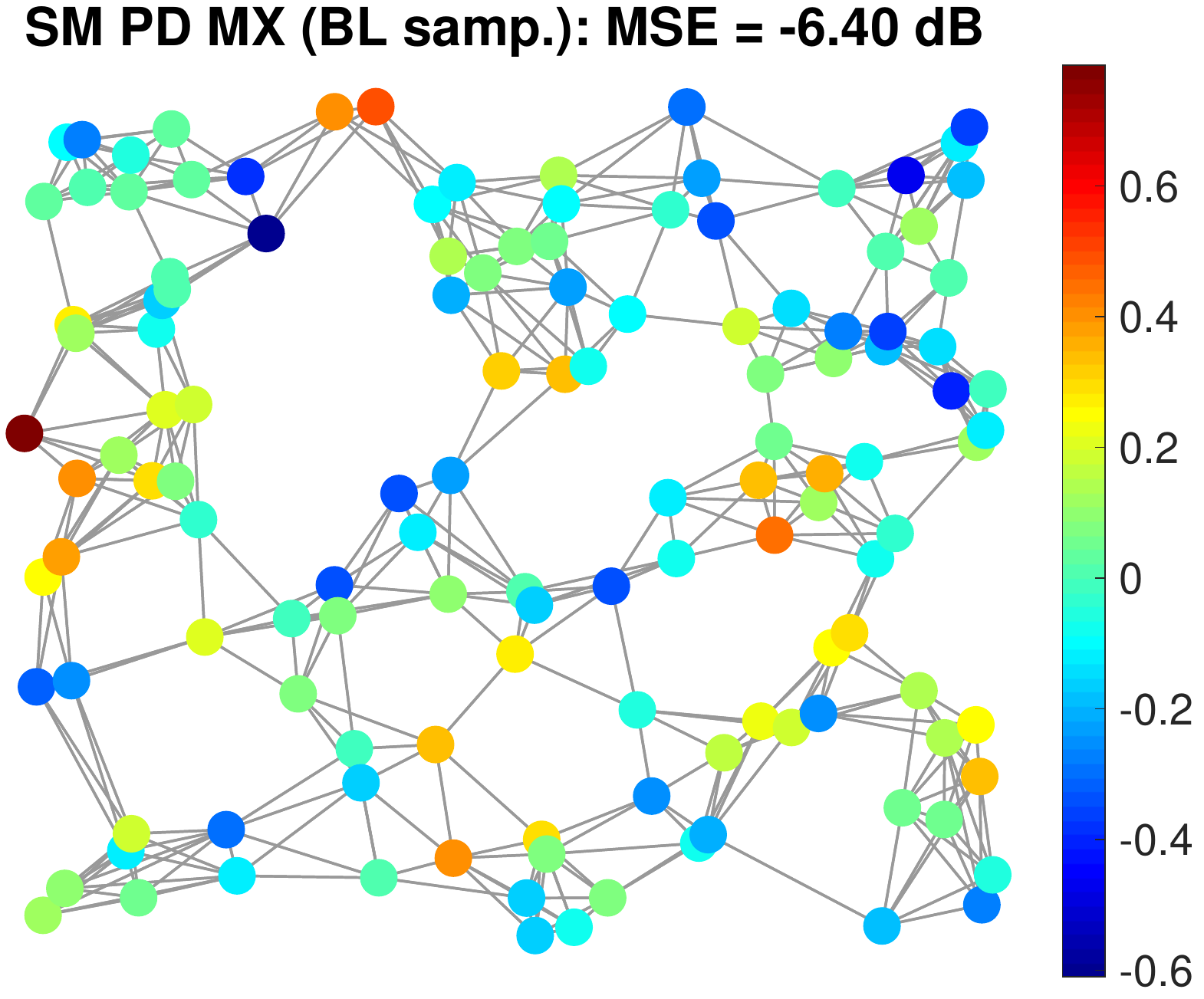}}\\
\subfigure[][Reconstructed: SS UNC (non-BL sampling)]{\includegraphics[width=.19\linewidth]{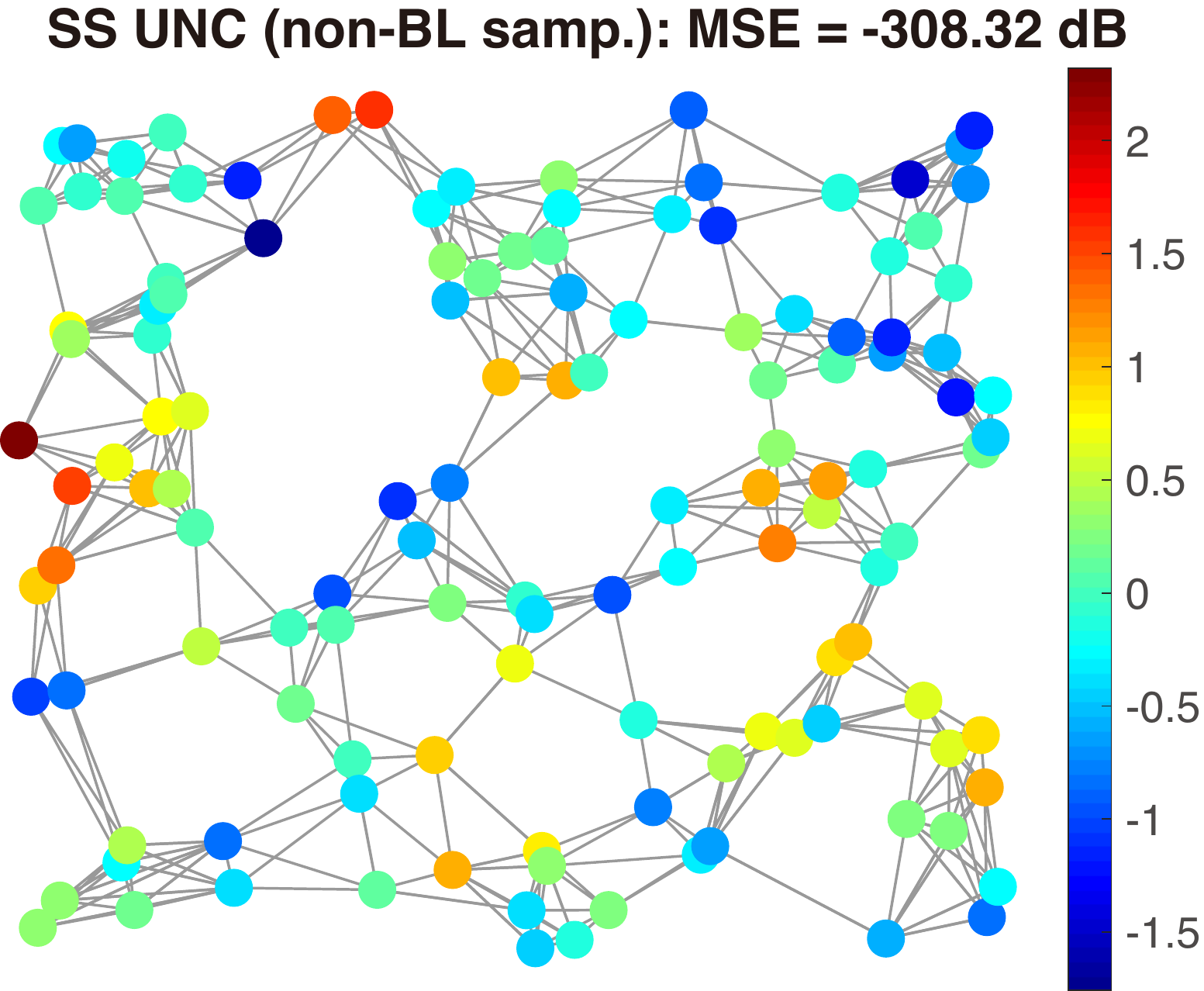}}\ 
\subfigure[][Reconstructed: SS PD DS/MX (non-BL sampling)]{\includegraphics[width=.19\linewidth]{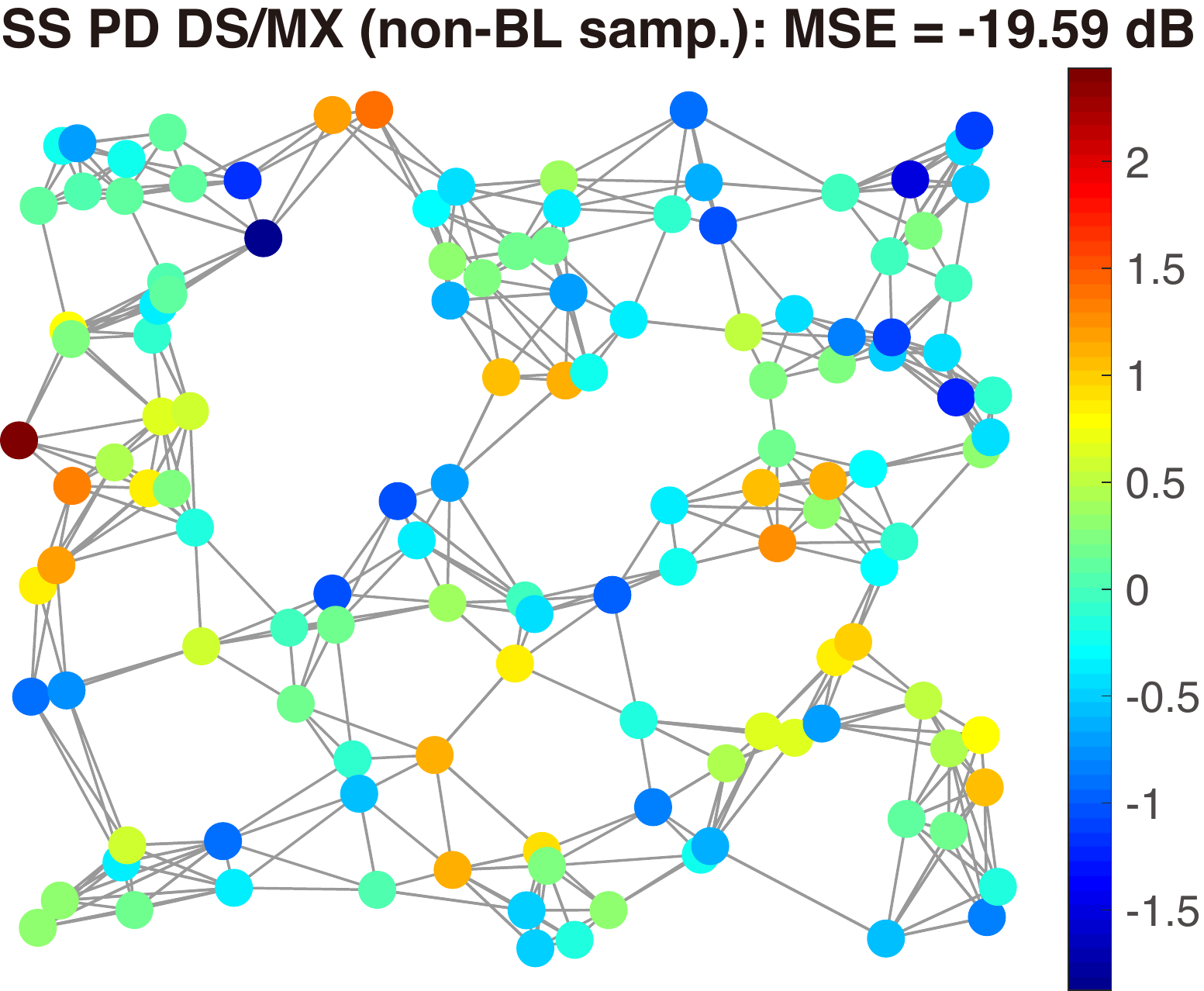}}\ 
\subfigure[][Reconstructed: SS PD LS (non-BL sampling)]{\includegraphics[width=.19\linewidth]{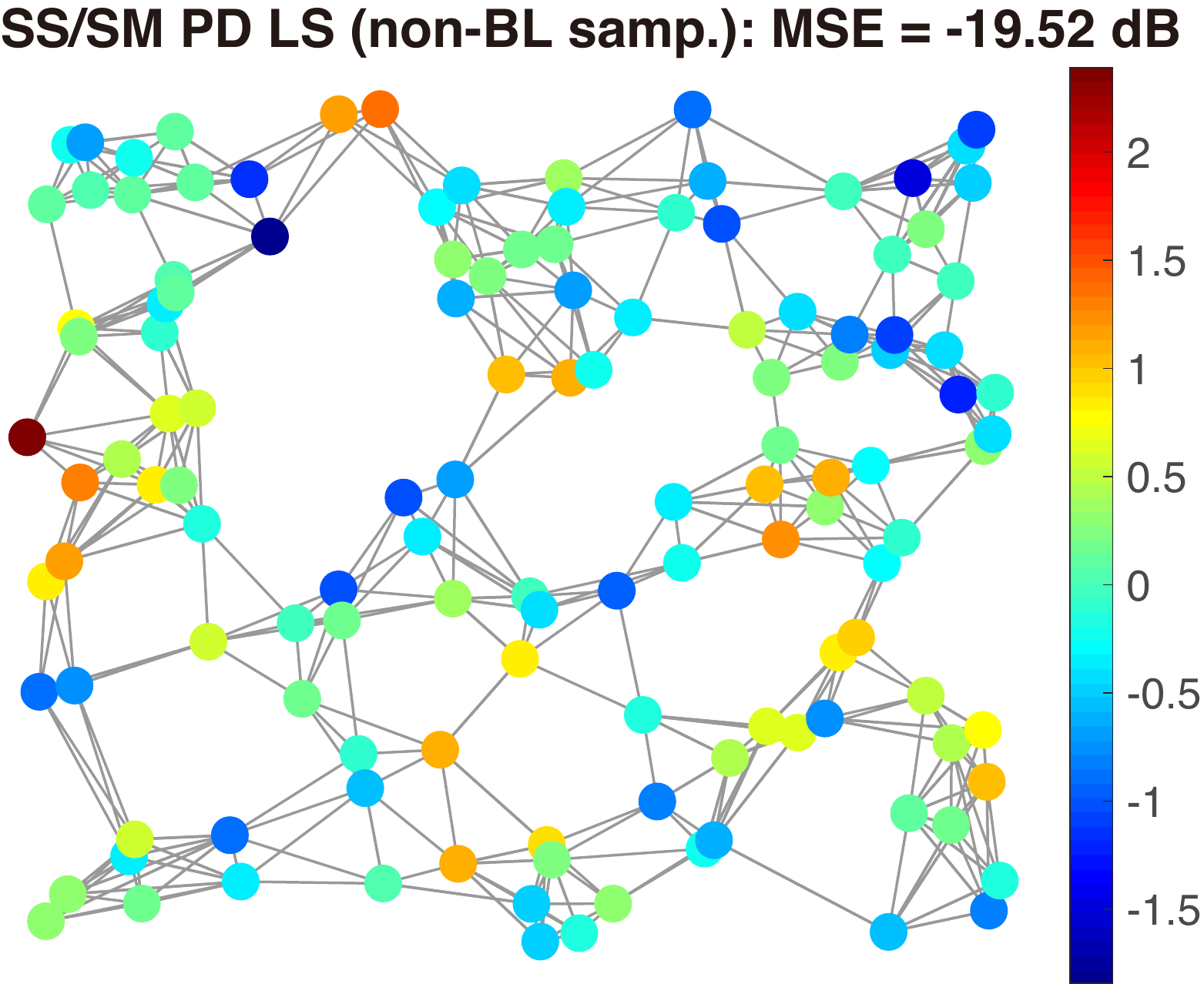}}\ 
\subfigure[][Reconstructed: SM UNC (non-BL sampling)]{\includegraphics[width=.19\linewidth]{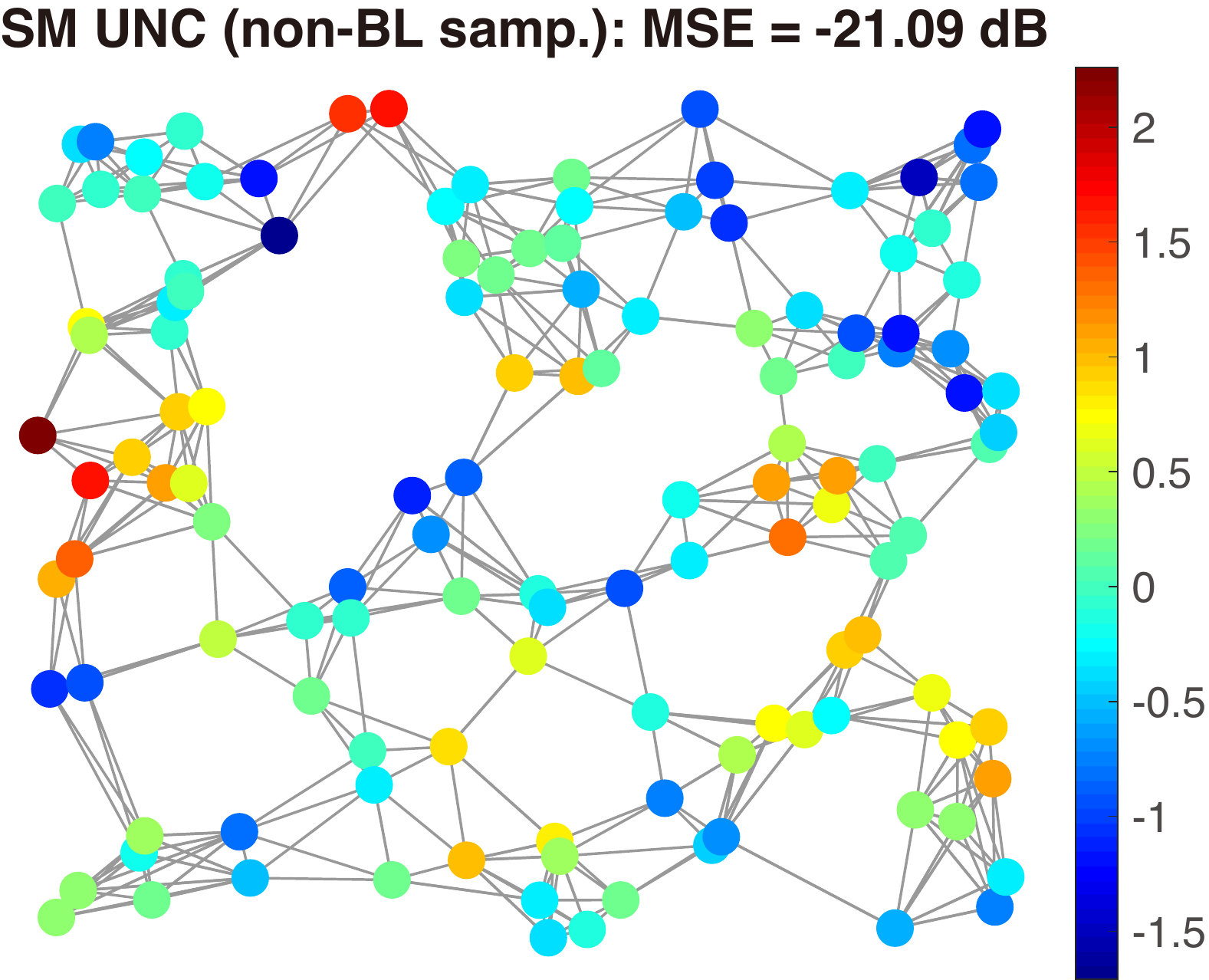}}\ 
\subfigure[][Reconstructed: SM PD MX (non-BL sampling)]{\includegraphics[width=.19\linewidth]{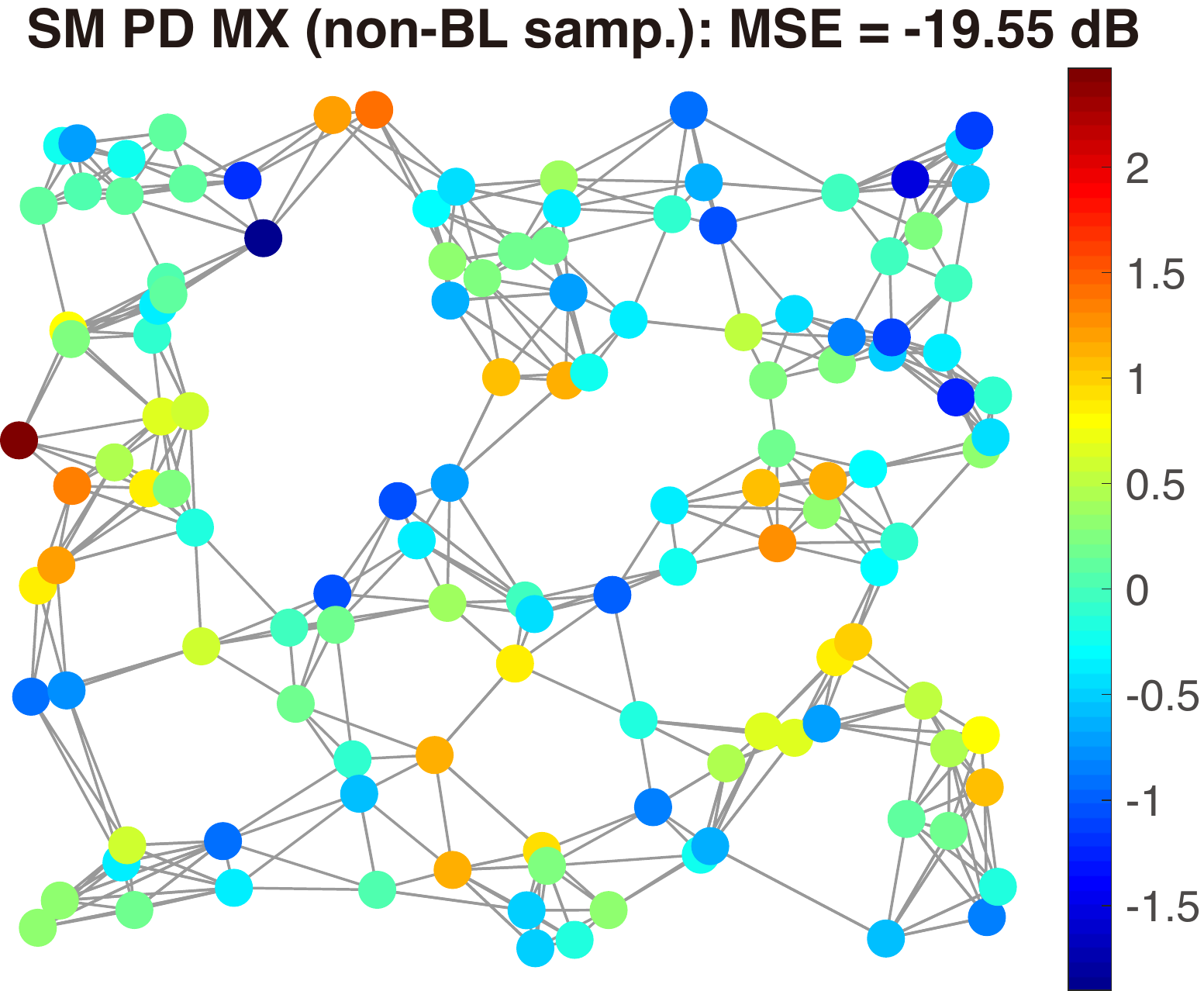}}
\caption{Signal recovery experiments for a random sensor graph. The generator function \#1 in \eqref{eqn:A_exp} is used. For visualization, we choose $N=128$ and $K=16$. Top row: The original signal and reconstructed signals with existing approaches. Middle row: Reconstructed signals with the proposed approach using bandlimited sampling. Bottom row: Reconstructed signals with the proposed approach using non-bandlimited sampling. SS and SM refer to the subspace and smoothness priors, UNC and PD refer to the unconstrained and predefined solutions.}
\label{fig:experiments_samp}
\end{figure*}

\begin{figure}[tp]
\centering
\subfigure[][Spectra for bandlimited sampling.]{\includegraphics[width=\linewidth]{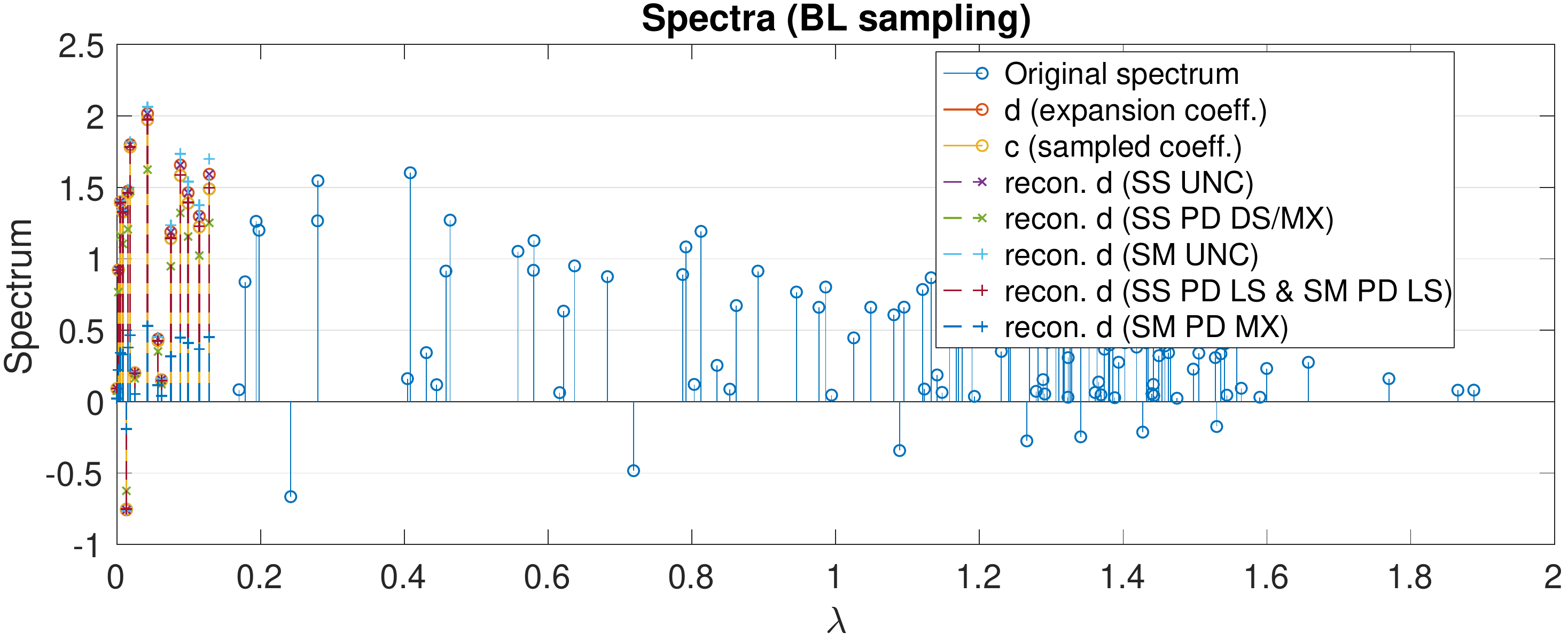}}\\
\subfigure[][Spectra for non-bandlimited sampling.]{\includegraphics[width=\linewidth]{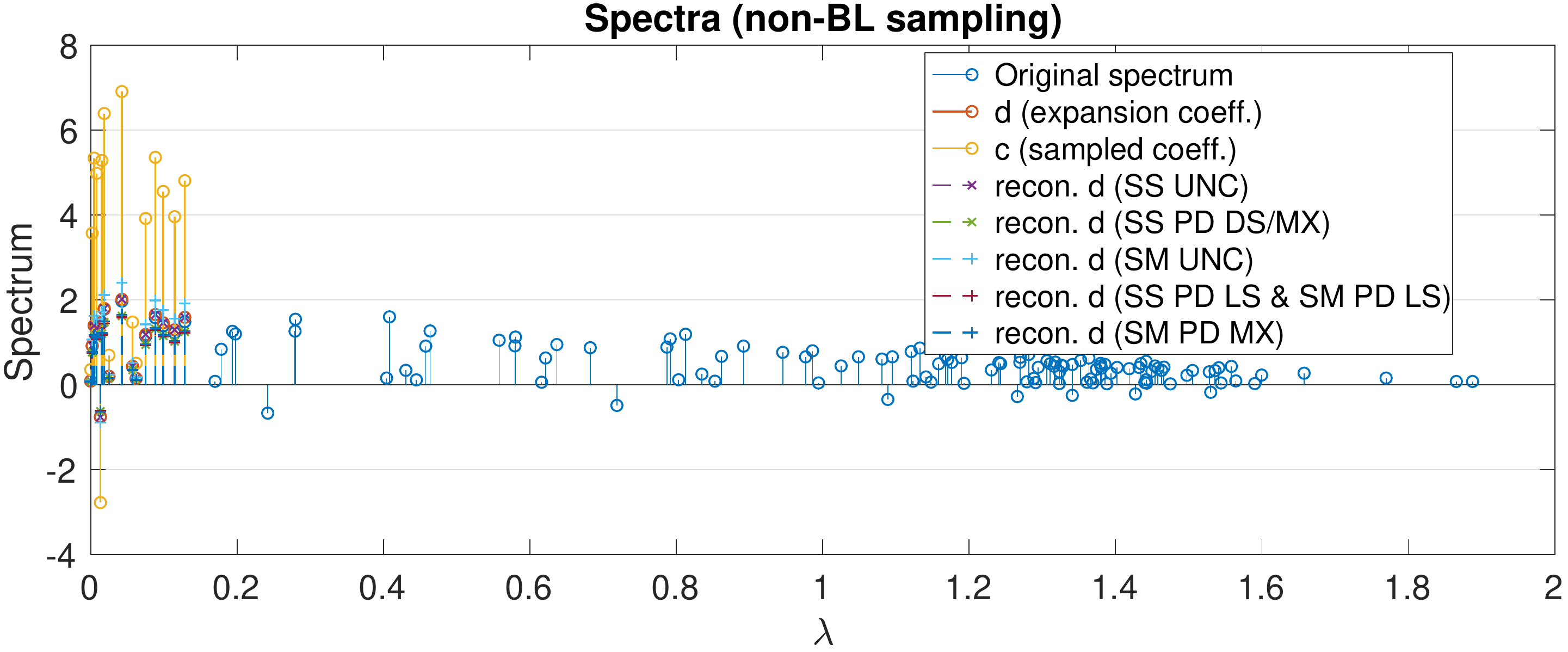}}\\
\caption{Original, sampled, and corrected spectra of Fig. \ref{fig:experiments_samp}.}
\label{fig:experiments_spectra}
\end{figure}

\begin{figure}[t]
   \centering
   \includegraphics[width = 0.8 \linewidth]{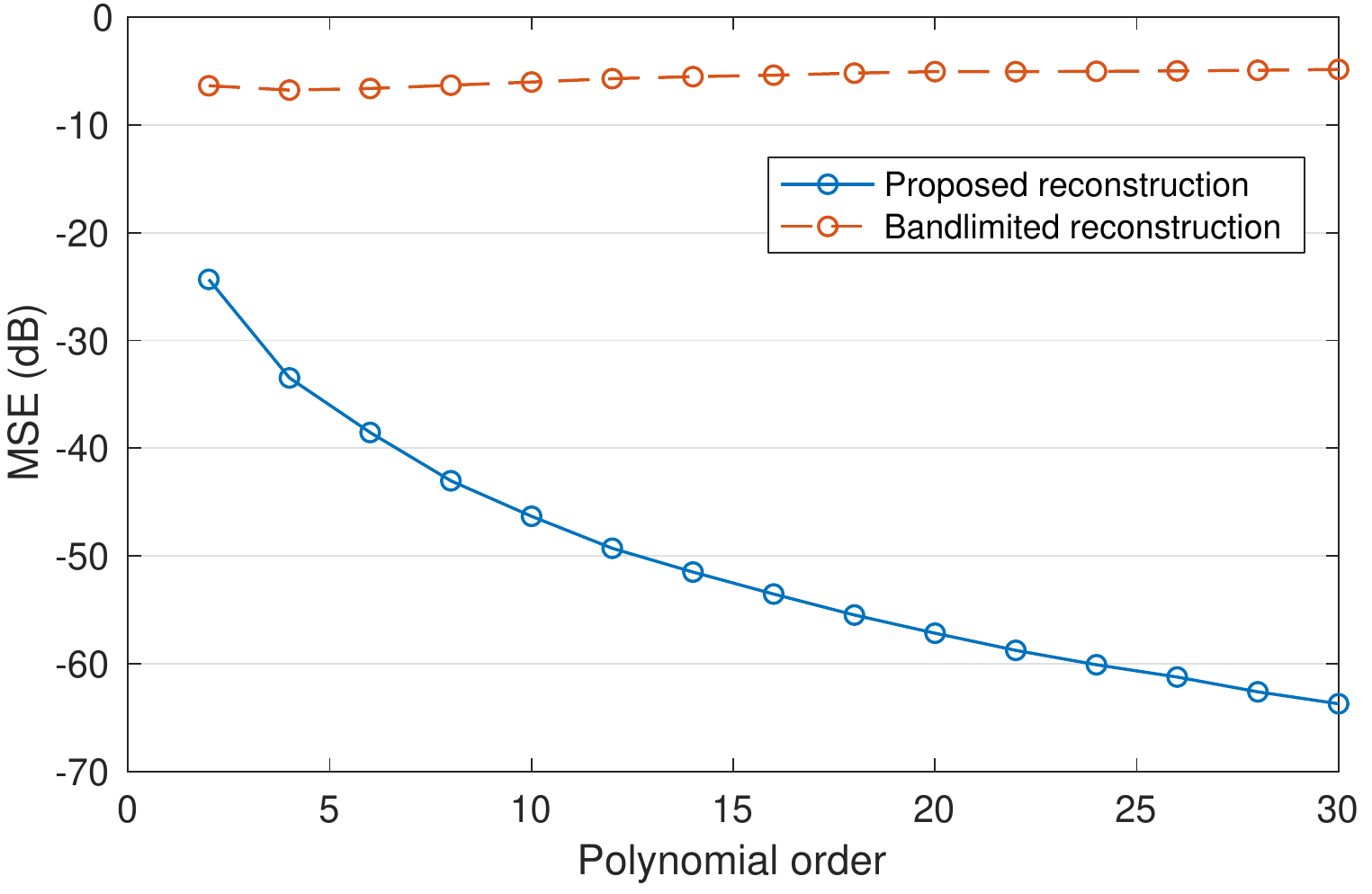} 
   \caption{MSE in recovery of signals on a bipartite graph. Here, $N = 256$ and the results are averaged over 100 independent runs.}
   \label{fig:bpt_mse}
\end{figure}

\begin{figure}[t]
\centering
\subfigure[][Original]{\includegraphics[width=.32\linewidth]{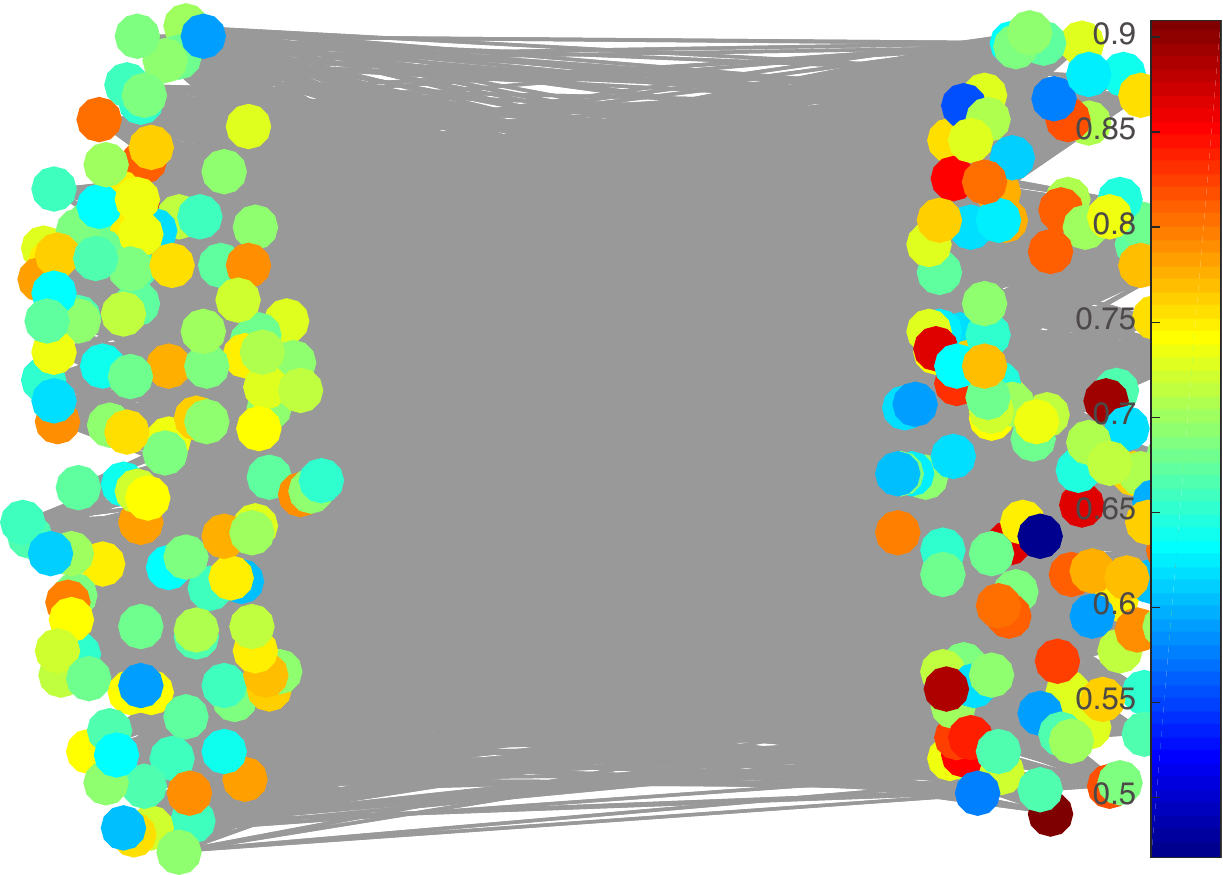}}
\subfigure[][Proposed reconstruction (MSE: $-54.45$ dB)]{\includegraphics[width=.32\linewidth]{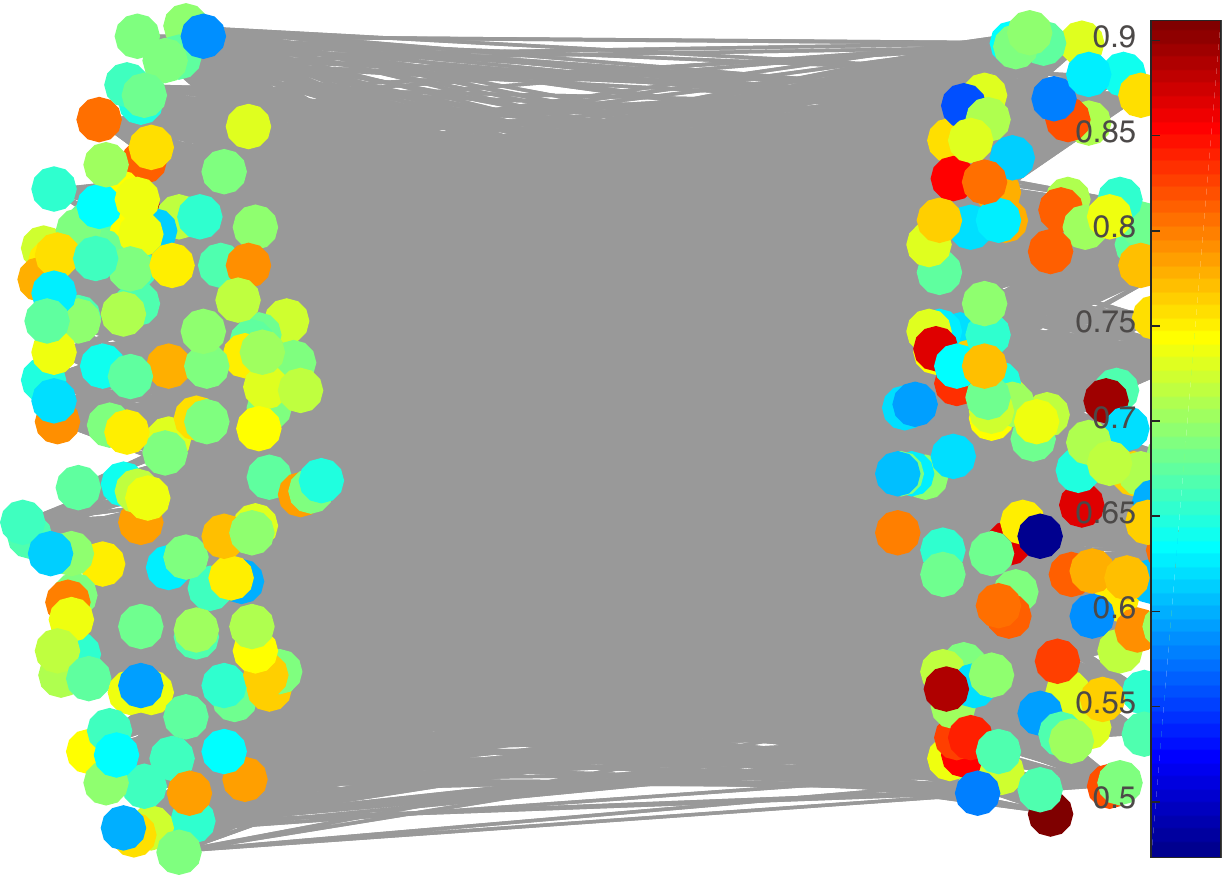}}
\subfigure[][Bandlimited reconstruction (MSE: $-25.95$ dB)]{\includegraphics[width=.32\linewidth]{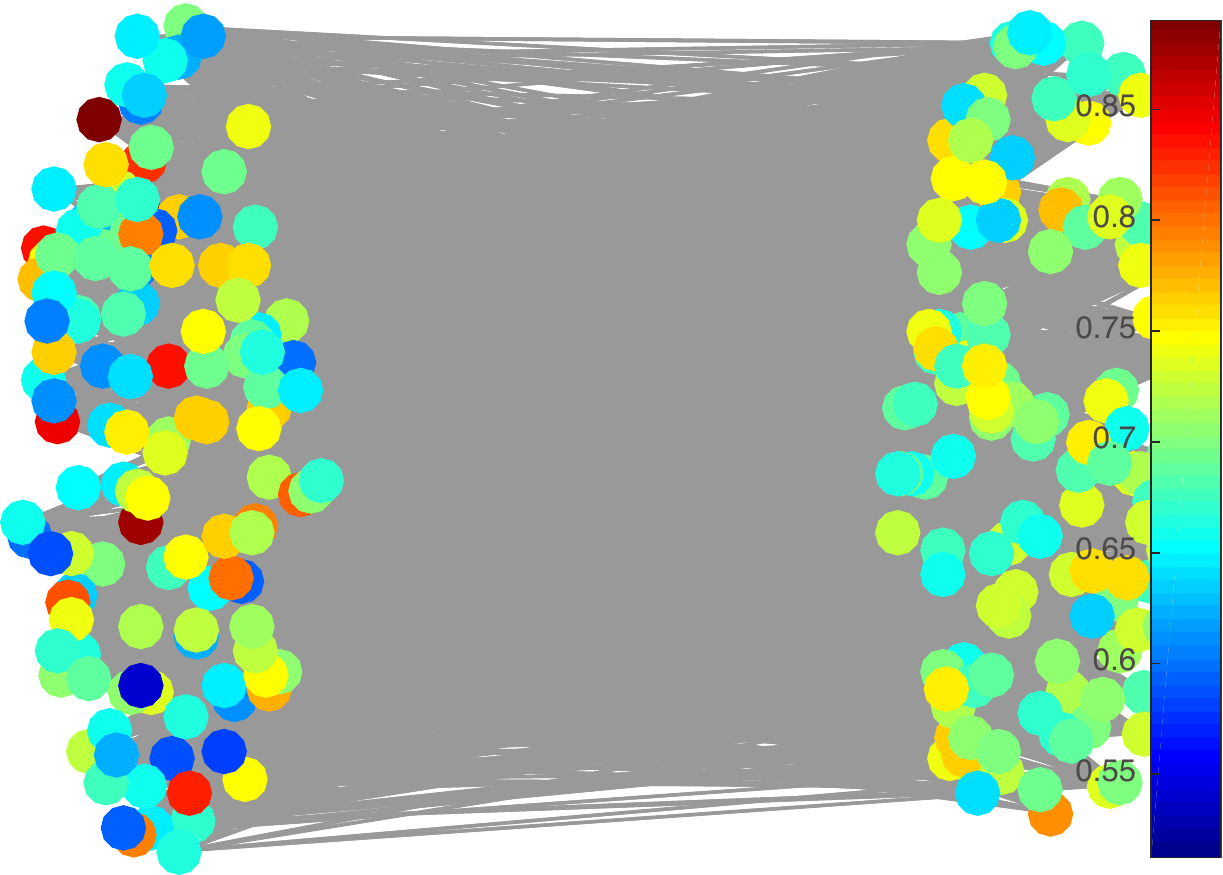}}
\caption{Original and reconstructed signals on a bipartite graph with $N = 256$. The right and left vertex sets correspond to $\mathcal{V}_1$ (retained) and $\mathcal{V}_2$ (discarded), respectively. The expansion coefficients $\bmd$ are drawn from $\mathcal{N}(0.25\times 10^{-2}, 1)$ for clear visualization. Chebyshev polynomial approximation of order $16$ is used both for $\widetilde{W}(\bLam)$ and $\widetilde{S}(\bLam)$.}
\label{fig:signals_bpt}
\end{figure}

Table \ref{tab:experiments} summarizes the average MSEs for various sampling approaches. The visualization of the reconstructed signals are shown in Fig. \ref{fig:experiments_samp}. Their corresponding spectra are also shown in Fig. \ref{fig:experiments_spectra}.

\textbf{Noiseless Signals: }
The unconstrained solution for the subspace prior perfectly recovers the original signal with machine precision. The predefined solutions for both the subspace and smoothness priors contain some reconstruction errors; however, in most cases, they are much smaller than those in the bandlimited reconstruction, especially for the generator function \#1 in \eqref{eqn:A_exp}. Because the signals generated by \eqref{eqn:A_exp2} are smoother than those of \eqref{eqn:A_exp}, the recovery with the smoothness prior is comparable to the bandlimited reconstruction if bandlimited sampling is applied. By contrast, using the non-bandlimited sampling function, all predefined solutions are superior to the bandlimited sampling/reconstruction. The unconstrained solution with a smoothness prior \eqref{eqn:smooth_unc} yields the same results as those in bandlimited sampling and reconstruction when using $G_{\text{BL},K}(\li)$ as the sampling filter.

FastGSSS uses vertex domain sampling and assumes smoothness (or bandlimitedness) of the signals. Its MSE is slightly worse than the predefined solutions and comparable to the bandlimited sampling/reconstruction using spectral domain sampling. This is because the signals used in the experiment are not fully bandlimited. As shown in Fig. \ref{fig:experiments_samp}(c), its reconstructed signal is over-smoothed compared to the original signal in Fig. \ref{fig:experiments_samp}(a). In fact, the sampling set selection strategy of FastGSSS (and any other reconstruction method utilizing vertex domain-based sampling) is based on the assumption that the signal is sufficiently smooth, i.e., a special case of the smoothness prior described in Section \ref{sec:smoothness}. This leads to existing recovery methods based on graph sampling theory interpolating missing values with a smooth graph filter, which thus have difficulty using the recovery of non-bandlimited signals even when $M$ is increased. As mentioned in Section \ref{subsec:complexity}, recovery based on vertex domain sampling generally requires a matrix inversion whose computational complexity is typically $O(K^3)$, and the entire reconstruction matrix needs to be re-calculated even when the sampling set $\mT$ is slightly changed. In contrast, our reconstruction matrix only differs in a diagonal matrix $\bH$ even when the sampling rate or sampling filter is changed. Therefore, $\bH$ can be easily re-calculated, as long as sampling is applied to the same graph.

\textbf{Noisy Signals: }
All methods contain increased errors for noisy cases, as expected. The unconstrained solution for the subspace prior demonstrated a significantly worse performance than that of a noiseless case because it did not assume any smoothness of the reconstructed graph signals. The predefined filters, both with subspace and smoothness priors, have demonstrated a performance close to those of unconstrained solutions because their reconstruction filters yield a smooth signal. 

Bandlimited reconstruction occasionally outperforms generalized sampling because it removes the noise in the high-graph-frequency band. Therefore, the bandlimited reconstruction for the smoother signals in \eqref{eqn:A_exp2} is comparable to generalized sampling with non-bandlimited sampling. However, for wider-band signals, as in \eqref{eqn:A_exp}, generalized sampling is much better than a bandlimited method even for noisy cases.

FastGSSS is inferior to spectral domain sampling-based approaches in most cases as in the experiment for noiseless signals. Because such a smoothing filter is naturally robust to noise, MSEs of FastGSSS are stabilized for noisy cases (but are still inferior to GFT domain sampling).

\subsection{Recovery Experiment on Bipartite Graphs with Vertex Domain Sampling}\label{subsec:exp_bpt}
We next demonstrate the recovery of non-bandlimited graph signals from vertex domain sampling, as described in Sections \ref{subsec:bpt} and \ref{subsec:subspace_unc}. To the best of our knowledge, this example is the first attempt to recover full-band graph signals from vertex domain operations without utilizing a multi-band decomposition.

In the signal recovery of this experiment, we set $S(\bLam) = G_{\text{BL}, N/2}(\bLam)$ and $A(\li) = W(\li) = G_{\text{IR}}(\li)$, as described in Section \ref{subsubsec:ss_unc_bpt}. That is, the original signal is full-band, whereas the sampled signal is low-pass filtered, i.e., the high-graph-frequency components are discarded after sampling. We assume that the generator information is losslessly available for recovery. The sampling and recovery of this situation is formulated in \eqref{eqn:recon_bpt}. If graph filters $\mathbf{G}$ and $\bW'$ can be represented as vertex domain filters, all operations will be performed in the vertex domain, i.e., without applying GFT (or an eigendecomposition of the variation operator). The filters, i.e., $S(\bLam)$ and $W'(\bLam)$ in \eqref{eqn:recon_bpt}, cannot be represented as vertex domain filters in general. Fortunately, they can be approximated as vertex domain operators by utilizing polynomial approximations of the spectral filter responses. In this experiment, we use a Chebyshev polynomial approximation (CPA) \cite{Hammon2011, Shuman2011}. As described in Section \ref{sec:SGT}, the $P$th order a CPA of an arbitrary graph spectral filter corresponds to a vertex domain filter with $P$-hop localization.

The original signal $\boldx$ is obtained as follows\footnote{Note that the generation process in \eqref{eqn:bpt_gen} uses the non-polynomial $\bW'$.}:
\begin{equation}
\label{eqn:bpt_gen}
\boldx = \mathbf{W}' \bI_{\mathcal{V}_1}^\top
\bmd,
\end{equation}
where each element in $\bmd$ is a random variable drawn from a normal distribution $\mathcal{N}(1, 1)$.
Figure \ref{fig:bpt_mse} shows the average MSEs of the reconstructed signals after $100$ independent runs according to the polynomial order. For comparison, we also plot the MSE of bandlimited reconstruction in which we use $\widetilde{S}(\li) = \widetilde{G}_{\text{BL}, N/2}(\li)$ as the reconstruction filter, of which $\widetilde{\cdot}$ denotes the polynomial approximated filter. As shown in Fig.~\ref{fig:bpt_mse}, the reconstruction error decreases monotonically as $P$ increases. The reconstructed signals are also shown in Fig. \ref{fig:signals_bpt}. The bandlimited reconstruction yields large errors, whereas the proposed reconstruction exhibits extremely similar signal values as the original values.

Future studies on this type of generalized sampling, particularly for the non-bipartite case, is an interesting topic for future research.

\section{Conclusion}\label{sec:conclusion}
We proposed a framework for generalized sampling of graph signals. We assumed that graph signals lie in a PGS subspace, which extends the SI subspace in standard signal processing to the graph setting. Sampling is defined in the graph frequency domain. We considered two priors for the graph signals, subspace and smoothness priors, which are parallel to those studied for signals in SI subspaces. All filters used in our framework can be represented as graph spectral filters. Numerical experiments demonstrated that our proposed sampling can recover a class of sampled signals broader than that obtained through existing graph sampling theories. We also presented perfect recovery of non-bandlimited graph signals on bipartite graphs without explicit operations in the GFT domain.


\end{document}